\documentclass{amsart}
\usepackage[foot]{amsaddr} 

\usepackage{geometry}
\usepackage{graphicx,psfrag,epsf}
\usepackage{enumerate}
\usepackage{enumitem}
\usepackage{amsfonts}
\usepackage{mathtools}
\usepackage{amssymb}
\usepackage{amsmath}
\allowdisplaybreaks
\usepackage{longtable}
\usepackage{bigints}
\usepackage{siunitx}
\usepackage{amsthm}
\usepackage{soul}
\usepackage{color}

\newcommand{\addition}[1]{#1}

\usepackage{hyperref}
\usepackage[capitalise]{cleveref}
\newtheorem{theorem}{Theorem}[section]

\newtheorem{lemma}[theorem]{Lemma}

\usepackage[numbers]{natbib}
\usepackage{url} 
\usepackage{doi}


\begin{document}

\title{A robust Bayesian analysis of variable selection under prior ignorance}
\author{Tathagata Basu$^1$}
\email{tathagatabasumaths@gmail.com}
\author{Matthias C.~M.~Troffaes$^2$}
\email{matthias.troffaes@durham.ac.uk}
\author{Jochen Einbeck$^{2,3}$}
\email{jochen.einbeck@durham.ac.uk}
\address{$^1$UMR CNRS 7253 Heudiasyc, Universit\'{e} de Technologie de
Compi\`{e}gne}
\address{$^2$Department of Mathematical Sciences, Durham University}
\address{$^3$Durham Research Methods Centre}

\begin{abstract}
We propose a \addition{cautious} Bayesian variable selection routine \addition{by investigating} the
\addition{sensitivity of a} hierarchical model,
where the regression coefficients \addition{are specified by} spike and slab priors.
We exploit the use of latent variables to understand the importance of the co-variates.
These latent variables also allow us to obtain the size of the model space which
is an important aspect of high dimensional problems. 
\addition{In our approach,}
instead of fixing a single prior, we adopt a \addition{specific type of} robust
Bayesian \addition{analysis, where we consider a set of priors within the same parametric family}
to specify the selection probabilities of these latent
variables. \addition{We achieve that by considering} a set of expected \addition{prior} selection
probabilities, \addition{which allows us to perform a sensitivity analysis}
to understand the effect of prior elicitation on the variable
selection. The sensitivity analysis provides us sets of posteriors for the
regression coefficients \addition{as well as the selection indicators}
and \addition{we show
that the posterior odds of the model selection probabilities are monotone
with respect to the prior expectations of the selection probabilities.
We also} analyse synthetic and real life datasets
to illustrate our \addition{cautious variable selection method and compare
it with other well known methods}. 
\end{abstract}

\maketitle

\section{Introduction}

High dimensional modelling is a key issue in modern science and technology. 
In a regressional context, we consider a problem to be high dimensional if the 
number of co-variates present in the model is more than the total number of
observations. Let $y \coloneqq (y_1$, \dots, $y_n)^T$ denote the vector of
$n$ real valued responses and $\mathbf{x}\coloneqq$ $[\mathbf{x}_1$, \dots, $\mathbf{x}_n]^T$
denote the corresponding predictors where each $\mathbf{x}_i$ is a $p$-dimensional 
column vector. Then for a vector of regression coefficients $\beta\coloneqq(\beta_1$,
\dots, $\beta_p)^T$, we can define a linear model in the following way:
\begin{equation}\label{eq:lm}
    y = \mathbf{x}\beta + \epsilon
\end{equation}
where $\epsilon\coloneqq(\epsilon_1$, \dots, $\epsilon_n)^T$ is a vector of the
noises which are \addition{assumed to be} normally distributed. For high dimensional problems, $p>n$
leads to potential difficulties in the parameter estimation. Naturally,
we wish to perform a variable selection to overcome the issue.
That is, we want to estimate these regression coefficients so that only few
of them are non-zero (or, active). We aim to construct a Bayesian 
routine which is \addition{cautious} in variable selection and incorporates
prior information \addition{efficiently}.

Several works have been done on variable selection from both a frequentist and a
Bayesian point of view. The frequentist approaches are usually performed
by adding a penalty term to the log likelihood of a linear model. One such 
method is the least absolute shrinkage and selection operator or simply LASSO 
\cite{tib1996}. Despite being a popular method owing to its easy computation,
LASSO fails to satisfy several asymptotic properties for consistent variable selection.
This led to the development of several other methods for consistent variable selection. 
\citet{fan2001} investigated the  asymptotic properties for variable selection 
and introduced the smoothly clipped absolute deviation or simply SCAD. Later,
\citet{Zou2006} introduced the adaptive LASSO, a weighted version
of LASSO that gives asymptotically unbiased estimates. 

Variable selection problems are well investigated in a Bayesian context as well.
\citet{tib1996} suggested the use of the double exponential distribution as a natural
prior for the regression coefficients for variable selection. This led to several
Bayesian alternatives for LASSO. \citet{park2008} proposed a hierarchical model 
for the Bayesian LASSO. Lykou and Ntzoufras \cite{Lykou2013} developed a concept for 
specification  of the hyperparameters based on Bayes factors which evaluate the evidence for inclusion 
of the respective predictor variables. Later \citet{bhatta2015} proposed the Dirichlet 
LASSO using a global-local mixture of Gaussians to specify the regression 
coefficients.

\addition{In this paper,}
we are particularly interested in the spike and slab models
for variable selection. One of the earlier works on spike and slab prior specification
can be found in  \cite{george_bvs}. The authors used latent variables for variable
selection. Later \citet{Ishwaran_2005} formalised the notion of spike and
slab priors and proposed a continuous bimodal prior for hyper-variances 
to attain sparsity. They suggested the following formulation for spike and
slab priors:
\begin{align}
\beta\mid\sigma^2,z
&\sim N(\mathbf{0}_p,\sigma^2\mathbf{D}_{z})\\
\sigma^2 &\sim \pi_1\\
z&\sim \pi_2
\end{align}
where $\mathbf{D}_{z}$ is a $p\times p$ dimensional diagonal matrix such that 
the diagonal entries are $z\coloneqq\left(z_1,\cdots,z_p\right)$. The choice
of $\pi_1$ and $\pi_2$ ensures that these exclude
values of zero with probability 1. Later \citet{narisetty2014} proposed
a hierarchical framework based on this and provided
strong consistency properties for spike and slab models.

Most of the Bayesian variable selection methods are developed on the basis of
posterior contraction rates of the regression coefficients \addition{and} these 
contraction rates are often derived based on several assumptions on the 
design matrix and level of sparsity. However, high dimensional problems \addition{may not contain} the necessary information to perform a 
Bayesian analysis based on these assumptions. 
\addition{To overcome this, \citet{george_emp_bayes} suggested an
empirical Bayes approach to incorporate prior information for variable
selection. However, choosing a single prior based on this approach can also 
be problematic for high dimensional problems, as it might lead to
overfitting. To avoid that, we }
tackle this problem from a robust Bayesian \cite{BERGER1990303} point of
view \addition{and consider a set of priors, based on (multiple) prior elicitation(s). 
}

Robust Bayesian analysis was popularised by \citet{BERGER1990303}. In robust
Bayesian analysis, we consider a set of priors to capture prior
information in a careful manner so that it represents the prior uncertainty. 
The use of a set of priors results in a set
of posteriors instead of a single posterior. Several robust Bayesian 
\addition{approaches have been discussed} in the context of regression: 
robust Bayesian analysis for
linear regression by \citet{CHATURVEDI1996175}; the imprecise logit-normal model
by \citet{Bickis2009}; the multinomial logistic regression by \citet{Paton2015}, just to 
name a few. However, a robust Bayesian approach for high dimensional modelling is yet
to be proposed.

\addition{As hinted earlier, }in this article, we adopt a \addition{specific type of} robust Bayesian approach 
\addition{where we consider a set of priors within the same parametric family
and perform a sensitivity analysis over the possible values of the 
hyperparameters.} \addition{A convenient way to do such analysis is to specify} the level 
of sparsity through a set of expected
prior selection probabilities. For that, we consider \addition{a set of} beta
distributions for the prior selection probabilities and perform a sensitivity
analysis over the hyperparameters.
The sensitivity analysis allows us to understand the variability of the model
sparsity. We exploit the framework of \citet{narisetty2014}
to \addition{incorporate our sets of priors and perform the sensitivity analysis}.

The rest of the paper is organised as follows. In \cref{sec:model}, we describe
our hierarchical model for robust Bayesian analysis and discuss our motivation for
our choices of the hyperparameters.  \Cref{sec:post:orth} is focused on the
posterior computation for the orthogonal design case. We first discuss 
posterior distributions of the latent variables and a decision criterion for
variable selection. Next, we discuss the posterior distributions of the
regression coefficients along with their properties in the orthogonal design case. 
We then focus on the \addition{properties of the posteriors for the} general case 
and \addition{also} provide a Gibbs sampling framework for sampling from
posterior distributions in \cref{sec:post:general}. We illustrate our results
using both synthetic and real datasets in \cref{sec:illustrate}. Finally,
we conclude in \cref{sec:conc}.

\section{A Hierarchical Model}\label{sec:model} 
We start from the framework given by \citet{narisetty2014} to propose our
hierarchical model. Recall the linear model in \cref{eq:lm}. We assume 
$\epsilon_i\sim\mathcal{N}(0, \sigma^2)$ for $1\le i\le n$ to construct
our likelihood. Then for $1\le j \le p$ , we can formulate
the hierarchical model in the following way:
\begin{align}
    y\mid \beta, \sigma^2
    &\sim \mathcal{N}\left(\mathbf{x}\beta, \sigma^2\mathbf{I}_{n}\right)\\
    \beta_j\mid \gamma_j = 1, \sigma^2 
    &\sim \mathcal{N}(0, \sigma^2\tau_1^2)\label{eq:prior:beta:1}\\
    \beta_j\mid \gamma_j = 0, \sigma^2
    &\sim \mathcal{N}(0, \sigma^2\tau_0^2)\label{eq:prior:beta:0}\\
    \gamma_j\mid q_j &\sim \mathrm{Ber}(q_j)\\
    q_j &\sim \mathrm{Beta}(s\alpha_j, s (1 - \alpha_j))\\
    \sigma^2 &\sim \text{InvGamma}(a, b)
\end{align}
where $s,\alpha_j,a,b>0$ are fixed constants. We fix $\alpha_j$ in a sense
that this is not random in nature; however, we perform a sensitivity analysis
over this $\alpha_j$.

The latent variables $\gamma\coloneqq(\gamma_1,\cdots,\gamma_p)$ in the model correspond
to the spike and slab prior specification routine where each $\gamma_j$ acts
as a selection indicator for the $j$-th co-variate $\mathbf{x}_j$. We also fix
the scale parameters $\tau_0$ and $\tau_1$ to represent our spike and slab
model. We consider a sufficiently small value of $\tau_0$
$(1\gg\tau_0>0)$ so that  $\beta_j|\gamma_j=0$ has its probability mass 
concentrated around zero. Therefore the probability distribution of 
$\beta_j|\gamma_j=0$ represents the spike component of our prior specification. 
To construct the slab component, we consider $\tau_1$ to be large 
($\tau_1>1$). A large value of $\tau_1$ allows the probability distribution
to capture the non-zero effects of $\beta_j$. The scale parameter $\tau_1$ also 
allows us to express our prior belief about the plausible values of $\beta_j$. 

In our model, the tail of the distribution of $\beta_j$ is also dependent on the prior 
expectation of selection probability $\alpha_j$. 
Let $f_{\gamma_j}(\beta_j)$ be the density of $\beta_j \mid 
\gamma_j$ as mentioned in \cref{eq:prior:beta:1} and \cref{eq:prior:beta:0}. So,
\begin{equation}
    f_{\gamma_j}(\beta_j) \coloneqq \frac{1}{\sqrt{2\pi\sigma^2\tau_1^{2\gamma_j}\tau_0^{2(1-\gamma_j)}}}
    \exp\left(-\frac{\beta_j^2}{2\sigma^2\tau_1^{2\gamma_j}\tau_0^{2(1-\gamma_j)}}\right).
\end{equation}
Then the hierarchical model implies the following:
\begin{align}
    P(\beta_j\mid \sigma^2) 
    &= \sum_{\gamma_j}P(\beta_j\mid\gamma_j,\sigma^2)\left(\int 
    P(\gamma_j\mid q_j) P(q_j)dq_j\right)\\
    &=\sum_{\gamma_j}[f_{1}(\beta_j)]^{\gamma_j}
    [f_{0}(\beta_j)]^{1-\gamma_j}
    \left(\int q_j^{\gamma_j}(1-q_j)^{1-\gamma_j}P(q_j) dq_j\right)\\
    &=\sum_{\gamma_j}[\alpha_j f_{1}(\beta_j)]^{\gamma_j}
    [(1-\alpha_j)f_{0}(\beta_j)]^{1-\gamma_j}\\
    &= \alpha_jf_1(\beta_j)+ (1-\alpha_j)f_0(\beta_j).\label{eq:beta:prior:alpha}
\end{align}
\begin{figure}
    \centering
    \includegraphics[width = 0.75\linewidth]{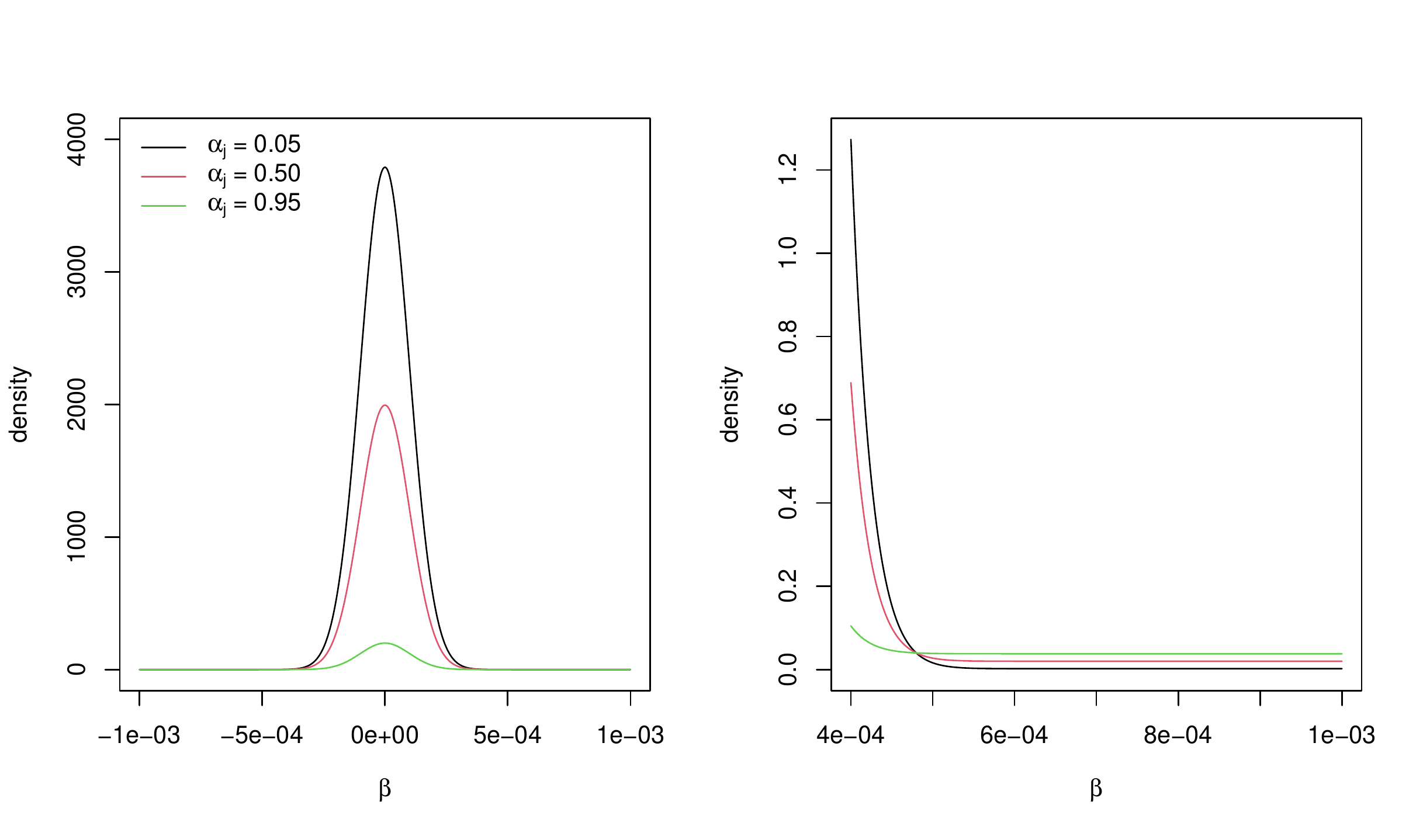}
    \caption{Marginalised densities of $\beta_j$ (\cref{eq:beta:prior:alpha}) for different
    values of $\alpha_j$. The figure on the right side shows the tails of the distributions.}
    \label{fig:prior:beta}
\end{figure}
That is, we can express our prior on $\beta_j$ as a mixture of normal distributions
where the weights are the prior expectation of the selection probability.
In \cref{fig:prior:beta} we show the effect of $\alpha_j$ on the prior
specification of $\beta$ for fixed $\tau_0=10^{-4}, \tau_1 = 10$ and $\sigma = 1$.
We notice that smaller values of $\alpha_j$ forces the prior 
to be more concentrated around 0 whereas higher values of $\alpha_j$
result to a flatter prior. This also suggests that we can impose our
prior belief on $\beta_j$ through $\alpha_j$. We can assign a sufficiently 
large value for $\tau_1$ to capture the prior expected range of $\beta_j$ and vary $\alpha_j$
to control the tail of the marginalised probability distribution.

\subsection{Robust Bayesian Analysis}
In this article, we are interested in a robust Bayesian analysis \cite{BERGER1990303}. 
In robust Bayesian analysis, we consider a set of priors instead of a single prior.
\addition{This set} can be constructed in different ways with focus on capturing all 
possible prior information efficiently. There are several \addition{applications of robust
Bayesian analysis}. One such use of robust Bayesian analysis for linear regression can be
found in \cite{CHATURVEDI1996175} where $\epsilon$-contaminated priors are used 
for the regression coefficients. In our case, we are interested in variable 
selection \addition{problems and therefore we require a setup which
allows us to select co-variates based on a robust decision rule.} 

From \cref{eq:beta:prior:alpha} and \cref{fig:prior:beta}, we observe the
effect of $\alpha_j$ in capturing the non-zero effects of $\beta$. 
Besides this, $\alpha_j$ reflects our prior information about the importance
of the $j$-th co-variate. However, for high dimensional problems, extracting
this information is difficult as the problem appears with very limited information.
This motivates us to perform a sensitivity analysis over $\alpha_j$ and
\addition{therefore,} we use a set of beta priors to specify the selection probability
$q_j$ such that for a fixed $s(>0)$, 
\begin{equation}
    \mathcal{Q}_j = \left\{\text{Beta}(s\alpha_j, 
    s(1-\alpha_j)) : \alpha_j\in \mathcal{P}_j\right)\}.
\end{equation}
The set $\mathcal{P}_j$ represents our prior information on the parameter $\alpha_j$
where $\mathcal{P}_j$ is any subset of $(0,1)$. This setting allows us to incorporate 
prior information in two different ways. We can consider specific $\mathcal{P}_j$
for individual $\alpha_j$ based on our prior information about the selection of
the $j$-th co-variate, or we may consider an equiprobable
setting where we assume $\alpha_1=\alpha_2=\cdots=\alpha_p$. 
For the equiprobable case, where we have no prior information about the problem, 
we may consider a near vacuous set for the
elicitation of each $\alpha_j$. That is, we consider
$\alpha_j\in [\epsilon_1, 1-\epsilon_2]$ where $1\gg\epsilon_1,\epsilon_2>0$.  
Alternatively, we can say that the prior expectation of the total number of active 
co-variates lies between $p\epsilon_1$ and $p(1-\epsilon_2)$.

\addition{\subsection{Co-variate Selection}\label{sec:co-var:sel}
The major aspect of the selection indicators $\gamma$ is to perform co-variate selection.
The selection indicators form a $2^p$ dimensional model space and ideally, we 
want to select the most probable model. However, this is extremely expensive
for high dimensional problems. \citet{george_bvs} suggested the use of median 
probability as a threshold for variable selection. That is, a co-variate is 
considered active if $P(\gamma_j\mid y) > 0.5$. Later \citet{barbieri2004}
showed that this threshold of 0.5 results in the optimal predictive model
and they refer to this model as median probability model.
The median probability for co-variate selection can be easily modified in terms
of the posterior odds and we can consider variables to be active if their posterior 
odds are greater than 1.

In our case, we have a set of posteriors for $\gamma_j$ instead of a single
posterior. Therefore, we need a slightly different decision criterion for 
co-variate selection. To propose a decision rule, we adapt the notion of median 
probability model with a stronger condition. We consider a co-variate to be 
inactive when 
\begin{align}
    \sup_{\alpha\in\mathcal{P}}
    \left\{\frac{P\left(\gamma_j = 1\mid y;\alpha\right)}
    {P\left(\gamma_j = 0\mid y\right)}\right\}
    & < 1, 
\end{align}
where $\mathcal{P}\coloneqq \mathcal{P}_1\times\cdots\times\mathcal{P}_p$. 
Similarly, we consider them active if,
\begin{align}
    \inf_{\alpha\in\mathcal{P}}
    \left\{\frac{P\left(\gamma_j = 1\mid y;\alpha\right)}
    {P\left(\gamma_j = 0\mid y\right)}\right\}
    & > 1. 
\end{align}
Note that here $\alpha$ is treated as a varying constant and not as 
a random variable.

Due to our stronger condition for variable selection, we may have some variables
which do not satisfy either of the above conditions. We call these variables
indeterminate variables. This way, we obtain a cautious variable selection 
paradigm. This also allows us to comment on the sensitivity of
the variables in our models using the posterior expectations as some of the variables will be
indeterminate which shows that their inclusion is dependent on the prior specification. 
}

\section{Posterior for Orthogonal Design}\label{sec:post:orth}
Our proposed hierarchical model allows us to obtain closed form expressions
for the posterior distributions of the regression coefficients and the
latent variables for the orthogonal design case that is when 
$\mathbf{x}^T\mathbf{x} = n \mathbf{I}_p$. For this assumption on the 
design matrix, we have $\hat{\beta} = \mathbf{x}^Ty/n$, where
$\hat{\beta}\coloneqq(\hat{\beta}_1$, \dots, $\hat{\beta}_p)^T$ are the
ordinary least squares estimates. Then,
\begin{align}
    P(y\mid \beta, \sigma^2)
    &=\frac{1}{\sqrt{(2\pi\sigma^2)^n}}
    \exp\left(-\frac{1}{2\sigma^2}\|y-\mathbf{x}\beta\|^2_2\right)\\
    &=\frac{1}{\sqrt{(2\pi\sigma^2)^n}}
    \exp\left(-\frac{1}{2\sigma^2}
    \left(n\beta^T\beta-2n\beta^T\hat{\beta} + y^Ty
    \right)\right)\\
    &\addition{=\frac{1}{\sqrt{(2\pi\sigma^2)^n}}
    \exp\left(-\frac{n\|\beta - \hat{\beta}\|_2^2}
    {2\sigma^2}\right)
    \exp\left(-\frac{y^Ty - n\hat{\beta}^T\hat{\beta}}
    {2\sigma^2}\right)}.\label{eq:likelihood}
\end{align}
Let $\gamma\coloneqq(\gamma_1$, \dots, $\gamma_p)$ and 
$q\coloneqq (q_1$, \dots, $q_p)$.
The joint posterior of the proposed hierarchical model can be computed
in the following way:
\begin{equation}
    P(\beta, \sigma^2, \gamma, q\mid y)\propto
    P(y\mid \beta, \sigma^2) P(\beta\mid\gamma, \sigma^2)
    P(\gamma\mid q)P(q)P(\sigma^2).\label{eq:joint:post}
\end{equation}
To analyse the properties of the posterior distributions of $\beta$ and $\gamma$,
we consider $\sigma^2$ to be known and fixed. These assumption on the design matrix
and variance allow us to show an interesting relationship between the posterior of the
selection indicators and the posterior of the regression coefficients. To show this
relationship, we first investigate the posterior of the latent variables and propose
a decision criterion for co-variate selection.

\subsection{Selection indicators}
Using \cref{eq:joint:post}, we write the posterior of $\gamma$ as
\begin{align}
    P(\gamma\mid y)
    &= \iint P(\beta, \gamma, q \mid y) dqd\beta\\
    &\overset{\gamma}{\propto} \int P(y\mid  \beta)
    \left(P(\beta\mid\gamma)\int P(\gamma\mid q)P(q) dq
    \right)d\beta.\label{eq:post:gamma}
\end{align}
\addition{Here, to avoid ambiguity we use the notation $\overset{\gamma}{\propto}$
which means that left hand side is equal to the right hand side up to a multiplicative 
constant which does not depend on $\gamma$. }

\addition{Now,} since $P(\gamma_j\mid q_j)=q_j^{\gamma_j}(1-q_j)^{1-\gamma_j}$ and $q_j$ follows a
$\mathrm{Beta}$ distribution,
\begin{align}
    &P(\beta\mid\gamma)\int P(\gamma\mid q)P(q) dq\nonumber\\
    &=\prod_{j}\left([f_{1}(\beta_j)]^{\gamma_j}
    [f_{0}(\beta_j)]^{1-\gamma_j}
    \int q_j^{\gamma_j}(1-q_j)^{1-\gamma_j}P(q_j) dq_j\right)\\
    &=\prod_{j}\left([\alpha_j f_{1}(\beta_j)]^{\gamma_j}
    [(1-\alpha_j)f_{0}(\beta_j)]^{1-\gamma_j}\right).
    \label{eq:joint:beta:gamma}
\end{align}
Combining \cref{eq:likelihood}, \cref{eq:post:gamma} and \cref{eq:joint:beta:gamma} 
we have,
\begin{align}
    P(\gamma\mid y)
    &\overset{\gamma}{\propto} \bigintssss \exp\left(-\frac{n\|\beta - \hat{\beta}\|_2^2}
    {2\sigma^2}\right)
    \prod_{j}\left([\alpha_j f_{1}(\beta_j)]^{\gamma_j}
    [(1-\alpha_j)f_{0}(\beta_j)]^{1-\gamma_j}\right)d\beta\\
    &\overset{\gamma}{\propto} \bigintssss
    \prod_{j}\left(\exp\left(-\frac{n(\beta_j - \hat{\beta}_j)^2}{2\sigma^2}\right)
    [\alpha_j f_{1}(\beta_j)]^{\gamma_j}
    [(1-\alpha_j)f_{0}(\beta_j)]^{1-\gamma_j}\right)d\beta\\
    &\overset{\gamma}{\propto} \prod_{j}\left(\bigintssss
    \exp\left(-\frac{n(\beta_j - \hat{\beta}_j)^2}{2\sigma^2}\right)
    [\alpha_j f_{1}(\beta_j)]^{\gamma_j}
    [(1-\alpha_j)f_{0}(\beta_j)]^{1-\gamma_j}d\beta_j\right)
\end{align}
This shows that the $\gamma_j$'s are a posteriori independent and we can write
the posterior of $\gamma_j$ in the following way:
\begin{align}
    P(\gamma_j\mid y)
    &=M_j\bigintssss 
    \exp\left(-\frac{n(\beta_j - \hat{\beta}_j)^2}{2\sigma^2}\right)
    [\alpha_j f_{1}(\beta_j)]^{\gamma_j}
    [(1-\alpha_j)f_{0}(\beta_j)]^{1-\gamma_j}d\beta_j,
\end{align}
where $M_j$ is a normalisation constant independent of $\gamma_j$. Then we have,
\begin{align}
    P(\gamma_j = 1\mid y) 
    & = M_j\alpha_j
    \bigintssss \exp\left(-\frac{n(\beta_j-\hat{\beta_j})^2}{2\sigma^2}\right)
    f_{1}(\beta_j)d\beta_j.
\end{align}
To simplify the above expression, we first propose the following lemma.
\begin{lemma}\label{lem:1}
For $k\in\{0,1\}$ and $j\in\{1,\cdots,p\}$ we have 
\begin{equation}
    \exp\left(-\frac{n(\beta_j-\hat{\beta_j})^2}{2\sigma^2}\right)f_{k}(\beta_j)
    =w_{k,j}\frac{1}{\sqrt{2\pi}\sigma_k}
    \exp\left(-\frac{\left(\beta_j - \hat{\beta}_{k,j}\right)^2}{2\sigma_k^2}\right)
    \label{eq:simplified}
\end{equation}
where, $\hat{\beta}_{k,j}\coloneqq\frac{n\tau_k^2\hat{\beta_j}}{n\tau_k^2 +1}$, 
$\sigma_k^2\coloneqq\frac{\sigma^2\tau_k^2}{n\tau_k^2 +1}$ and
$w_{k,j}\coloneqq \frac{1}{\sqrt{n\tau_k^2+1}}
\exp\left(-\frac{n\hat{\beta_j^2}}{2(n\sigma^2\tau_k^2+\sigma^2)}\right)$.
\end{lemma}
\begin{proof}
\addition{The proof is straightforward and therefore has been omitted.}
\end{proof}
Now, using \cref{lem:1} we have
\begin{equation}
    P(\gamma_j = 1\mid y) =M_j\alpha_jw_{1,j}
\end{equation}
and
\begin{equation}
    P(\gamma_j = 0\mid y) 
    = M_j(1-\alpha_j)w_{0,j}.
\end{equation}
Therefore, $\gamma_j$ follows a Bernoulli distribution such that,
\begin{equation}
    \gamma_j \mid y \sim
    Ber\left(\frac{\alpha_jw_{1,j}}{\alpha_jw_{1,j}+(1-\alpha_j)w_{0,j}}\right)
    \label{eq:post:gamma:3}.
\end{equation}

\subsubsection{Properties of the posterior}
\addition{From the co-variate selection rules defined in \cref{sec:co-var:sel}, 
we can show that a variable is inactive when,}
\begin{equation}
    \sup_{\alpha_j\in\mathcal{P}_j}
    \left\{\frac{w_{1,j}\alpha_j}
    {w_{0,j}(1-\alpha_j)}\right\} < 1.\label{eq:non:active}
\end{equation}
Similarly, we consider a co-variate to be active if,
\begin{equation}
    \inf_{\alpha_j\in\mathcal{P}_j}
    \left\{\frac{w_{1,j}\alpha_j}
    {w_{0,j}(1-\alpha_j)}\right\} > 1.\label{eq:active:strong}
\end{equation}

\addition{We can also see from \cref{eq:non:active} and \cref{eq:active:strong} that} 
the posterior odds are monotone with respect to $\alpha_j$ and
the posterior odds increase as we increase the value of $\alpha_j$.
Therfore, we only need to compute the posterior odds on the lower and upper 
limits of the set instead of the whole interval. \addition{For instance,
for the near vacuous case,}
\begin{align}
    \sup_{\alpha_j\in[\epsilon_1, 1-\epsilon_2]}
    \left\{\frac{w_{1,j}\alpha_j}
    {w_{0,j}(1-\alpha_j)}\right\} = 
    \frac{(1-\epsilon_2)}{\epsilon_2}\cdot \frac{w_{1,j}}{w_{0,j}}\\
    \intertext{and,}
    \inf_{\alpha_j\in[\epsilon_1, 1-\epsilon_2]}
    \left\{\frac{w_{1,j}\alpha_j}
    {w_{0,j}(1-\alpha_j)}\right\} = 
    \frac{\epsilon_1}{(1-\epsilon_1)}\cdot \frac{w_{1,j}}{w_{0,j}}.
\end{align}
Therefore, a co-variate is considered to be active if 
$\frac{\epsilon_1}{(1-\epsilon_1)}\cdot \frac{w_{1,j}}{w_{0,j}}>1$
and a co-variate is considered to be inactive if
$\frac{(1-\epsilon_2)}{\epsilon_2}\cdot \frac{w_{1,j}}{w_{0,j}}<1$.

\subsection{Regression coefficients}
The joint posterior of regression coefficients ie $\beta$ is given by:
\begin{align}
    P(\beta\mid y)
    &= \sum_{\gamma}\int P(\beta, \gamma, q \mid y) 
    dq\\
    &\overset{\beta}{\propto} \sum_{\gamma}\int P(y\mid \beta) 
    P(\beta\mid\gamma)P(\gamma\mid q)P(q) dq\\
    &\overset{\beta}{\propto} P(y\mid \beta)
    \sum_{\gamma}\left(P(\beta\mid\gamma)\int P(\gamma\mid q)P(q) dq
    \right).\label{eq:post:beta}
\end{align}
From \cref{eq:joint:beta:gamma} we have
\begin{align}
    P(\beta\mid\gamma)\int P(\gamma\mid q)P(q) dq
    =\prod_{j}\left([\alpha_j f_{1}(\beta_j)]^{\gamma_j}
    [(1-\alpha_j)f_{0}(\beta_j)]^{1-\gamma_j}\right).
\end{align}
Then we can write \cref{eq:post:beta} as
\begin{align}
    P(\beta\mid y)
    &\overset{\beta}{\propto} P(y\mid \beta)
    \sum_{\gamma}\left(\prod_{j}\left([\alpha_j f_{1}(\beta_j)]^{\gamma_j}
    [(1-\alpha_j)f_{0}(\beta_j)]^{1-\gamma_j}\right)
    \right)\nonumber.
\end{align}
Therefore swapping sum and product operations we get,
\begin{align}
    P(\beta\mid y)
    &\overset{\beta}{\propto} P(y\mid \beta)
    \prod_{j}\sum_{\gamma_j}\left([\alpha_j f_{1}(\beta_j)]^{\gamma_j}
    [(1-\alpha_j)f_{0}(\beta_j)]^{1-\gamma_j}\right)\\
    &\overset{\beta}{\propto} P(y\mid \beta)
    \prod_{j}[\alpha_jf_1(\beta_j)+ (1-\alpha_j)f_0(\beta_j)].
    \label{eq:post:beta:2}
\end{align}
Now combining \cref{eq:likelihood} and \cref{eq:post:beta:2} we have
\begin{align}
    P(\beta\mid y)
    &\overset{\beta}{\propto} \exp\left(-\frac{1}{2\sigma^2}
    \left(n\beta^T\beta-2n\beta^T\hat{\beta}\right)\right)
    \prod_{j}[\alpha_jf_1(\beta_j)+ (1-\alpha_j)f_0(\beta_j)]\nonumber\\
    &\overset{\beta}{\propto} \exp\left(-\frac{n}{2\sigma^2}
    \|\beta - \hat{\beta}\|_2^2\right)
    \prod_{j}[\alpha_jf_1(\beta_j)+ (1-\alpha_j)f_0(\beta_j)]\\
    &\overset{\beta}{\propto} \prod_{j}\exp\left(-\frac{n(\beta_j - \hat{\beta}_j)^2}{2\sigma^2}\right)
    [\alpha_jf_1(\beta_j)+ (1-\alpha_j)f_0(\beta_j)].\label{eq:post:beta:prod}
\end{align}
Therefore, the $\beta_j$'s are a posteriori independent and for
each $1\le j\le p$, we have,
\begin{align}
    P(\beta_j\mid y)
    &\overset{\beta_j}{\propto} \exp\left(-\frac{n(\beta_j - \hat{\beta}_j)^2}{2\sigma^2}\right)
    [\alpha_jf_1(\beta_j)+ (1-\alpha_j)f_0(\beta_j)].\label{eq:post:beta:j}
\end{align}
Let $W_j \coloneqq \alpha_jw_{1,j}+ (1-\alpha_j)w_{0,j}$. 
Then combining \cref{eq:post:beta:j} and \cref{eq:simplified} we have,
\begin{align}
    \beta_j \mid y
    &\sim\frac{\alpha_jw_{1,j}}{W_j}
    \mathcal{N}\left(\hat{\beta}_{1,j}, \sigma_1^2\right)
    +\frac{(1-\alpha_j)w_{0,j}}{W_j}
    \mathcal{N}\left(\hat{\beta}_{0,j}, \sigma_0^2\right).
    \label{eq:post:beta:mixture}
\end{align}
\begin{figure}
    \centering
    \includegraphics[width = 0.75\linewidth]{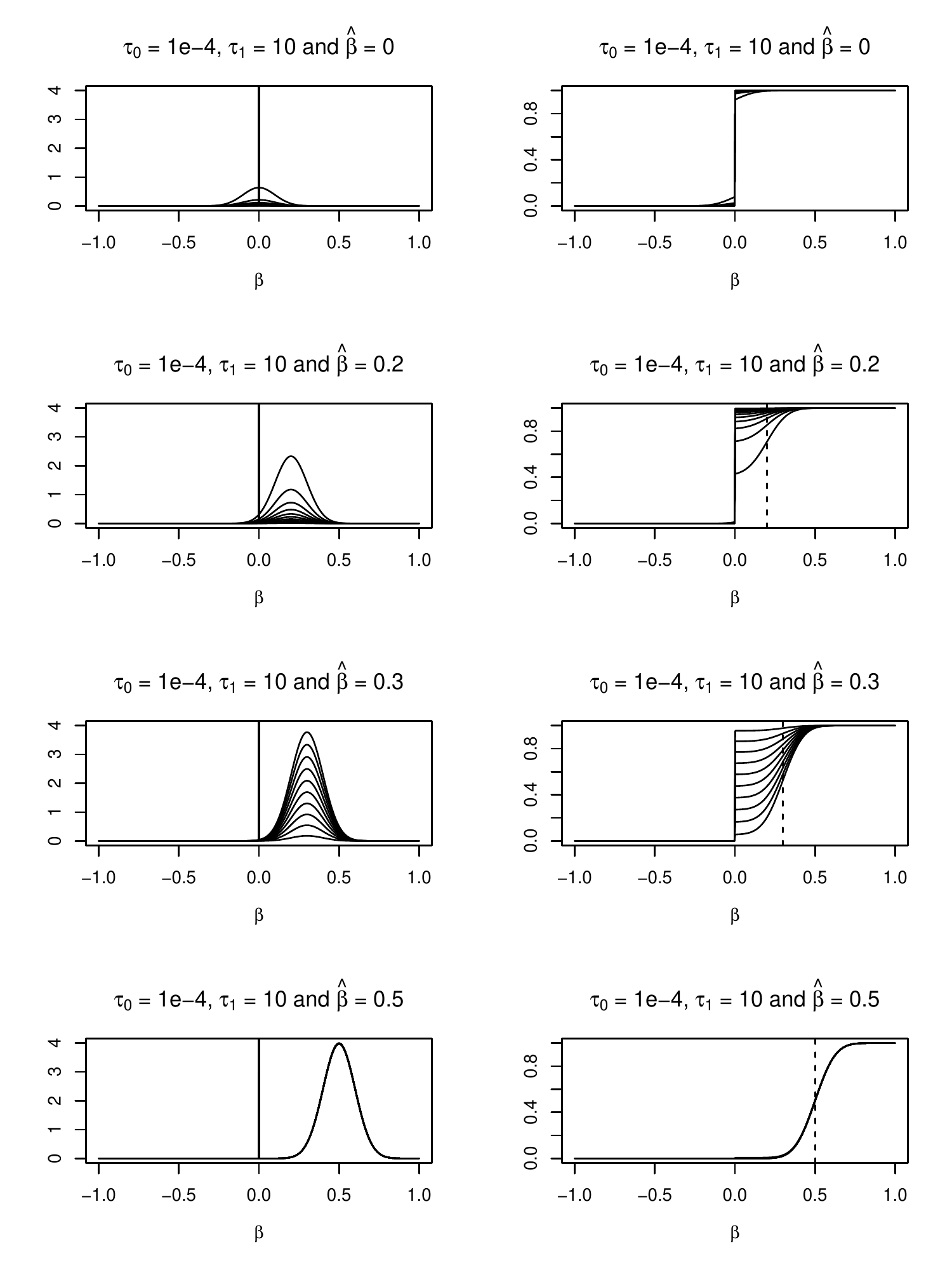}
    \caption{Posterior density function and corresponding cumulative distribution
    function of $\beta_j$ for different values of $\hat{\beta_j}$ over a set of $\alpha_j$ such that
    $\alpha_j\in [0.05, 0.95]$.}
    \label{fig:post:beta}
\end{figure}

\cref{eq:post:beta:mixture} shows that the posteriors of the regression coefficients 
are mixtures of two normal distributions. Clearly, the posteriors are bimodal when $\hat{\beta}_j\not=0$.
We illustrate the posteriors in \cref{fig:post:beta} for fixed $\sigma^2 = 1$,
$n=100$, $\tau_0=10^{-4}$ and $\tau_1$=10. In \cref{fig:post:beta}, the left column
shows the density functions and the right column shows the posterior cumulative
distribution functions (CDF). We show these posteriors
for four different values of $\hat{\beta}_j$ ($\hat{\beta}$ in the figure) over equispaced
grids of $\alpha_j$ so that $\alpha_j\in [0.05, 0.95]$. 
We observe that the posterior densities are bimodal except for the top row and
each of the posteriors has a spike component at zero. 
We also notice that for smaller values of 
$\hat{\beta}_j$, the posterior CDFs are more concentrated at zero. However, as we 
increase the value of $\hat{\beta}_j$, the posterior CDFs shift towards $\hat{\beta}_j$.
For a sufficiently large value of $\hat{\beta}_j$, the posterior CDFs  are concentrated at 
$\hat{\beta}_j$. 

\subsubsection{Properties of the posterior}
The posteriors of the regression coefficients enjoy several nice properties.
The posterior expectation of $\beta_j$
is given by
\begin{equation}
    E(\beta_j\mid y) = \frac{\alpha_jw_{1,j}}{W_j}\hat{\beta}_{1,j}
    +\frac{(1-\alpha_j)w_{0,j}}{W_j}\hat{\beta}_{0,j}.
\end{equation}
We observe that the posterior mean is monotonically increasing with respect to 
$\alpha_j$ for $\hat{\beta}_j > 0$ and monotonically decreasing with respect to
$\alpha_j$ for $\hat{\beta}_j < 0$ (see \cref{ap:lem:3} and the discussion following it). 

We show this in \cref{fig:post:exp:beta}. We
fix $n=100, \tau_0 = 10^{-4}, \tau_1 = 10$ and $\sigma^2 = 1$. We check posterior means
for six different possible values of $\hat{\beta}_j$. In the top row we show our results
for $\hat{\beta}_j>0$. We see that in the first two cases, the posterior means are 
monotonically increasing and in the third case it is close to constant. Similarly
in the bottom row, we show our result for $\hat{\beta}_j<0$. We see similarly that
the posterior means are decreasing in the first two cases, and remains close to
constant in the third case.

\begin{figure}[ht!]
    \centering
    \includegraphics[width = 0.95\linewidth]{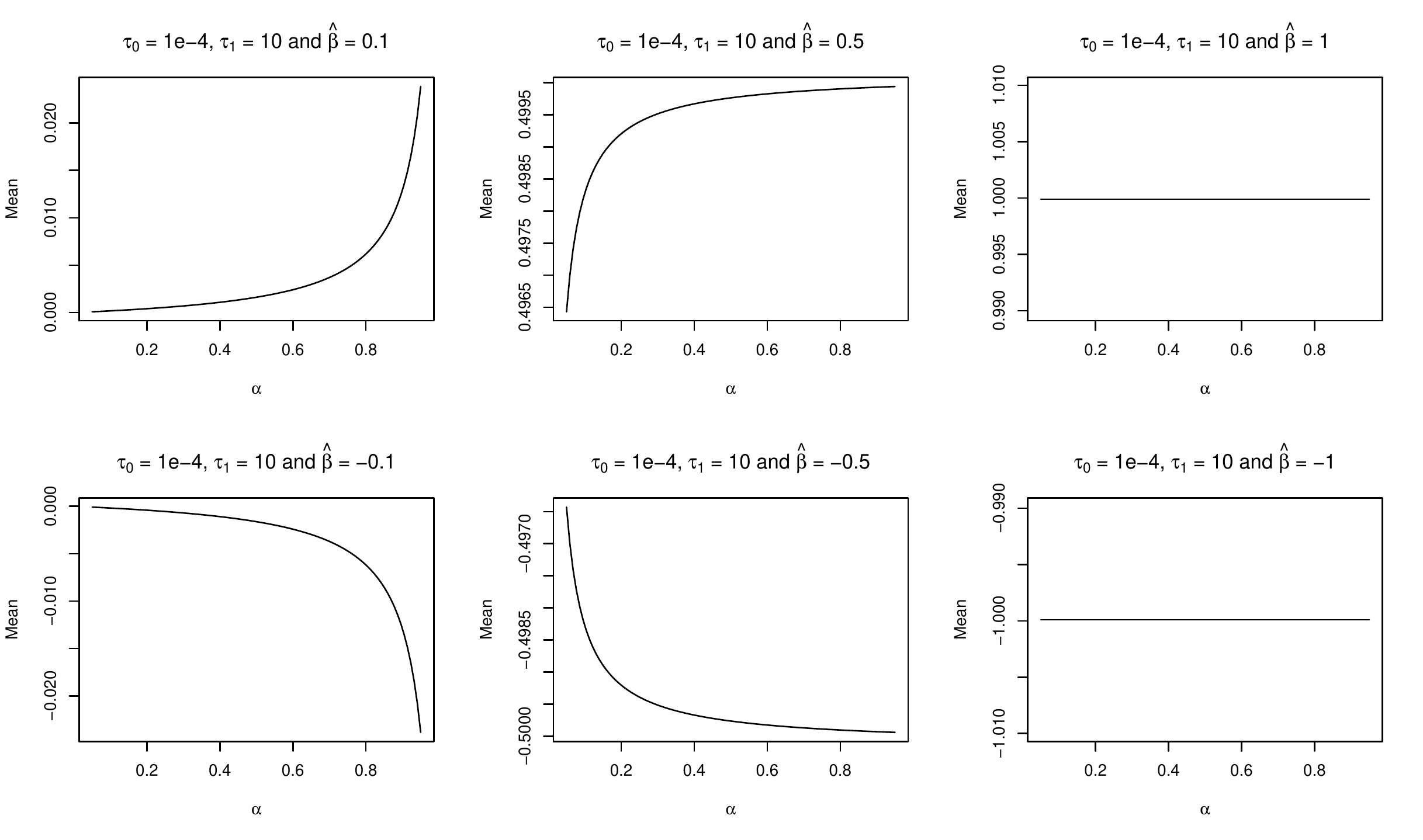}
    \caption{Relation between posterior expectation of $\beta$ and prior 
    selection probability $\alpha$ for different values of $\hat{\beta}$.}
    \label{fig:post:exp:beta}
\end{figure}

We also get a closed form expression for the posterior variance of $\beta_j$
(see \cref{ap:lem:4}), which is given by:
\begin{equation}
    \text{Var}(\beta_j\mid y) 
    = \frac{\alpha_jw_{1,j}\sigma^2_1 + (1-\alpha_j)w_{0,j}\sigma^2_0}{W_j}
    + \frac{\alpha(1-\alpha)w_{1,j}w_{0,j}(\hat{\beta}_{1,j} - \hat{\beta}_{0,j})^2}
    {W_j^2}.
\end{equation}
Therefore, we get a set of posterior variances $\mathcal{S}_j$ such that:
\begin{equation}
    \mathcal{S}_j\coloneqq \left\{\frac{\alpha_jw_{1,j}\sigma^2_1 + 
    (1-\alpha_j)w_{0,j}\sigma^2_0}{W_j}
    + \frac{\alpha(1-\alpha)w_{1,j}w_{0,j}(\hat{\beta}_{1,j} - \hat{\beta}_{0,j})^2}
    {W_j^2}
    :\alpha_j\in(0,1)\right\}
\end{equation}
where, $w_{k, j}$ and $\sigma_k$ are as defined before. The posterior variance of
$\beta_j$ does not have a monotonicity property like the posterior mean. 
In \cref{fig:beta:post:variance}, we show the effect of
$\alpha_j$ on the posterior variance for different values $\hat{\beta}$. We notice
that for extreme values of $\hat{\beta}$, the posterior variance is close to 
constant similar to our experience for posterior mean. 
\begin{figure}[ht!]
    \centering
    \includegraphics[width = 0.95\linewidth]{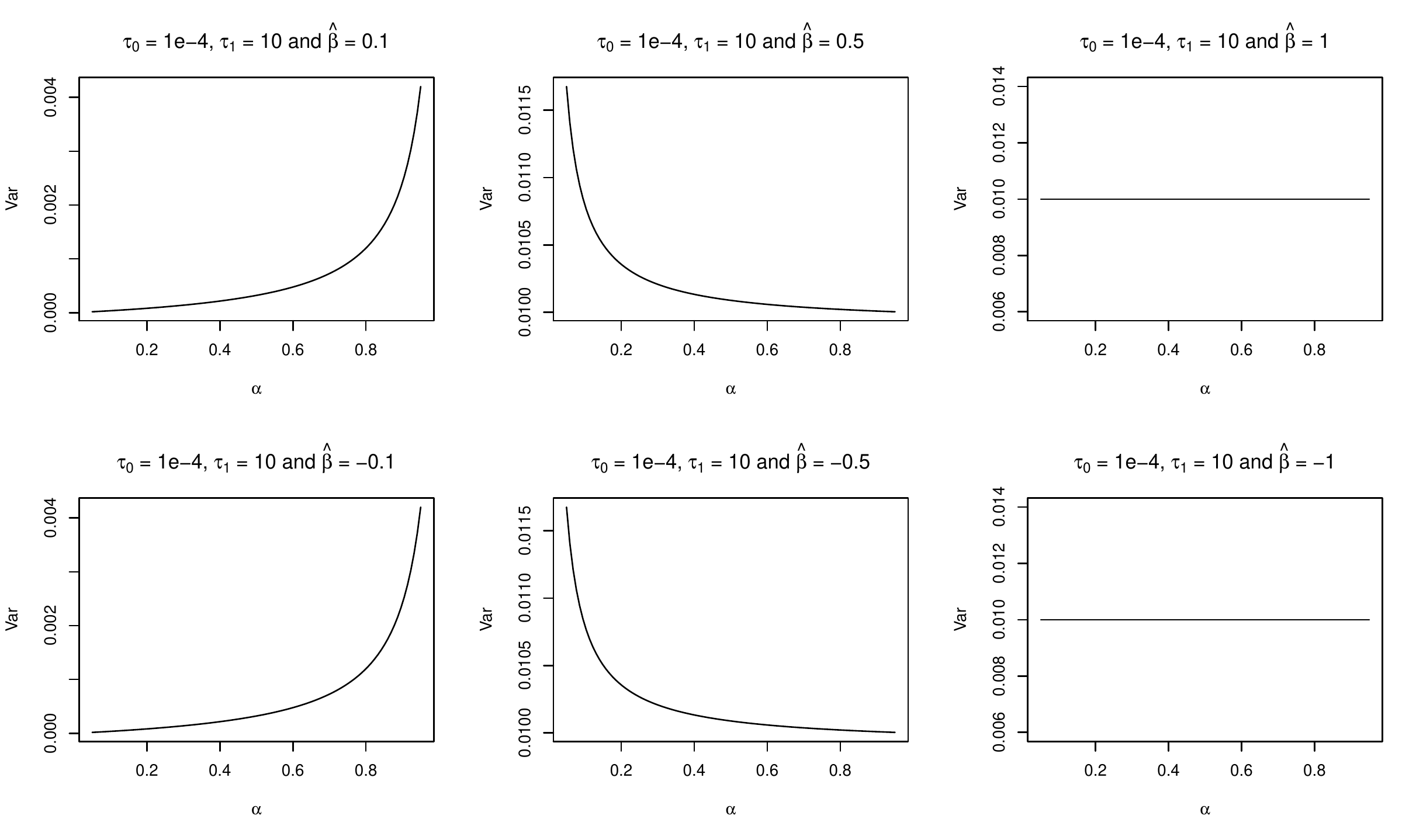}
    \caption{Relation between posterior variance of $\beta$ and prior 
    selection probability $\alpha$ for different values of $\hat{\beta}$.}
    \label{fig:beta:post:variance}
\end{figure}

\paragraph{Role in co-variate selection}
To show the role in co-variate selection, we first consider the ratios of the weights
in \cref{eq:post:beta:mixture}. For $1\le j \le p$, these ratios are given by:
\begin{equation}
    \frac{\alpha_j w_{1,j}}{(1-\alpha_j)w_{0,j}}.
\end{equation}
These ratios correspond to posterior selection probabilities of the selection indicators.
Therefore, for an active co-variate this ratio becomes greater than 1 for all 
$\alpha_j\in[\epsilon_1, 1-\epsilon_2]$ and $\mathcal{N}\left(\hat{\beta}_{1,j}, \sigma_1^2\right)$
dominates the posterior. Similarly, for a inactive co-variate this ratio becomes less 
than 1 for all values of $\alpha_j$ and
$\mathcal{N}\left(\hat{\beta}_{0,j}, \sigma_0^2\right)$ dominates the posterior. 
This allows us to propose the following lemma.

\begin{lemma}\label{lem:2}
$\mathcal{N}\left(\hat{\beta}_{1,j}, \sigma_1^2\right)$ dominates the posterior if
\begin{align}
    \hat{\beta_j^2} & >
    \frac{\sigma^2}{n}\frac{(n\tau_1^2+1)(n\tau_0^2+1)}{n\tau_1^2-n\tau_0^2}
    \left[
    2\ln\left(\frac{1-\epsilon_1}{\epsilon_1}\right)
    +\ln\left(\frac{n\tau_1^2+1}{n\tau_0^2+1}\right)\right].
\end{align}
and $\mathcal{N}\left(\hat{\beta}_{0,j}, \sigma_0^2\right)$ 
dominates the posterior if, 
\begin{align}
    \hat{\beta_j^2} & <
    \frac{\sigma^2}{n}\frac{(n\tau_1^2+1)(n\tau_0^2+1)}{n\tau_1^2-n\tau_0^2}
    \left[
    2\ln\left(\frac{\epsilon_2}{1-\epsilon_2}\right)
    +\ln\left(\frac{n\tau_1^2+1}{n\tau_0^2+1}\right)\right].
\end{align}
\end{lemma}

\begin{proof}
The proof of this lemma is provided in \cref{ap:lem:2}.
\end{proof}
\begin{figure}[ht]
    \centering
    \includegraphics[width = 0.95\linewidth]{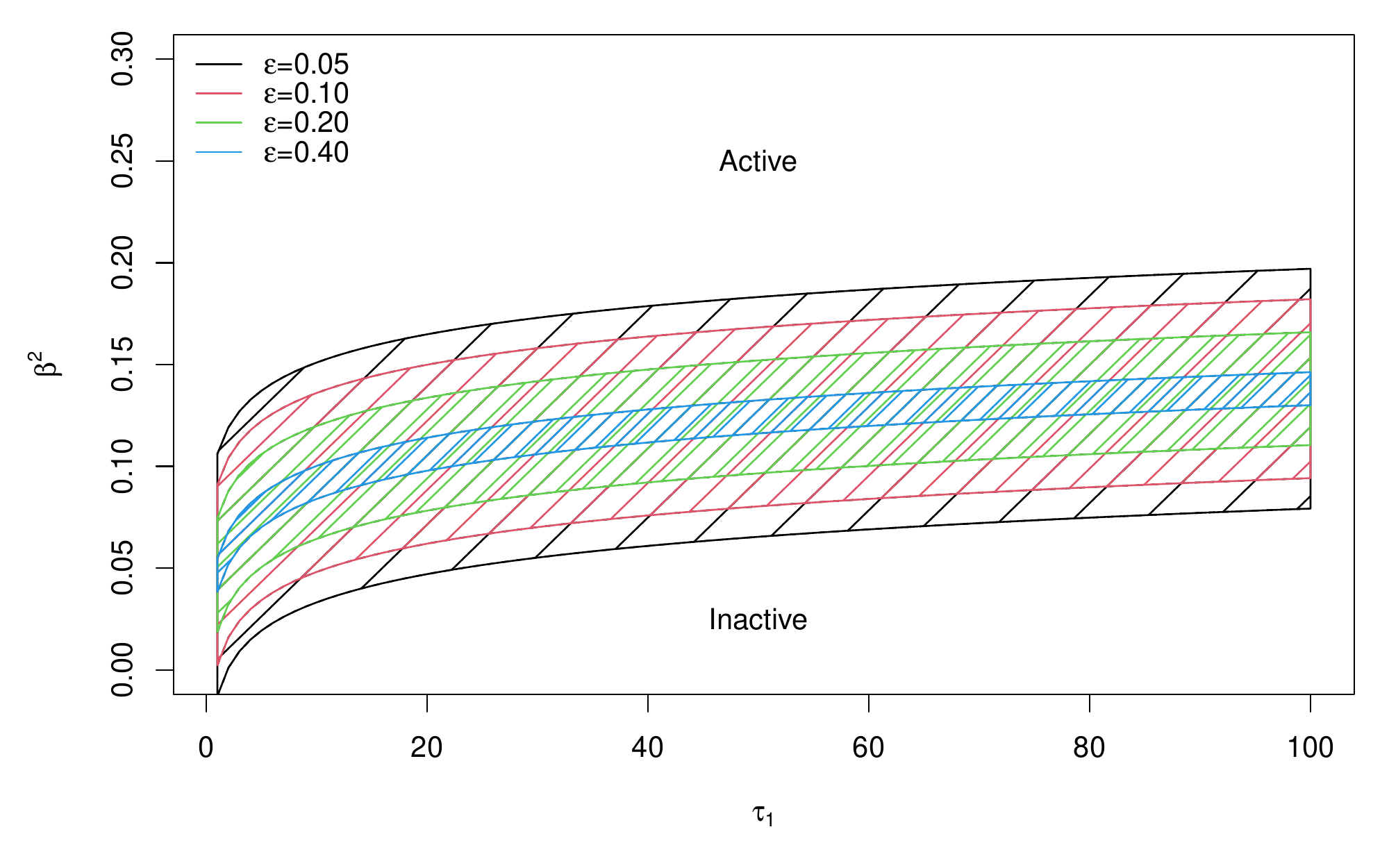}
    \caption{Effect of $\tau_1$ in specifying the region of indeterminacy
    for different values of $\epsilon$.}
    \label{fig:tau1:indet}
\end{figure}
We can further simplify this result if $\epsilon_1=\epsilon_2=\epsilon$, that is 
when $\alpha_j\in [\epsilon, 1 - \epsilon]$ and $\tau_0 \ll 1/n$, then 
$\mathcal{N}\left(\hat{\beta}_{1,j}, \sigma_1^2\right)$ dominates the posterior
if, 
\begin{equation}
    \hat{\beta}_j^2 > \frac{\sigma^2}{n}
    \frac{(n\tau_1^2+1)}{n\tau_1^2}\left[
    \ln\left(n\tau_1^2+1\right)
    +2\ln\left(\frac{1-\epsilon}{\epsilon}\right)
    \right],\label{post:condition:beta:post:low}
\end{equation}
and similarly, $\mathcal{N}\left(\hat{\beta}_{0,j}, \sigma_0^2\right)$ dominates the 
posterior if,
\begin{equation}
    \hat{\beta}_j^2 < \frac{\sigma^2}{n}
    \frac{(n\tau_1^2+1)}{n\tau_1^2}\left[
    \ln\left(n\tau_1^2+1\right)
    -2\ln\left(\frac{1-\epsilon}{\epsilon}\right)
    \right].\label{post:condition:beta}
\end{equation}
Clearly, for $\epsilon=0.5$, the right hand sides of \cref{post:condition:beta:post:low}
and \cref{post:condition:beta} are equal.

We can compute a region of indeterminacy using \cref{post:condition:beta:post:low}
and \cref{post:condition:beta}. If the value of $\hat{\beta}_j^2$ lies in between these
bounds then we consider the $j$-th   co-variate as indeterminate. We
illustrate this in \cref{fig:tau1:indet}. The shaded area shows the region
of indeterminacy for different values of $\alpha_j\in[\epsilon, 1-\epsilon]$.
Clearly, the region of indeterminacy depends on the values of $\epsilon$ and higher
values of $\epsilon$ shrink the region of indeterminacy. We also notice
that extreme values of $\tau_1$ may lead to poor results in variable selection. A very
small value of $\tau_1$ will force some inactive co-variates to be indeterminate
whereas a very high value of $\tau_1$ will force some non-zero small effects
to be inactive.

\section{Posterior Computation for General Case}\label{sec:post:general}

\addition{
The orthogonal case allows us to decompose the joint density function in a 
convenient way for known variance $\sigma^2$. However, in reality, 
the variance is unknown. Moreover, variable selection
is generally applied for correlated datasets
which cannot be transformed to an orthogonal design. 
As a consequence, in most practical cases, $\gamma_j$ and $\beta_j$ are 
no longer a posteriori independent. 
Instead, we have closed form expressions for the joint posteriors of $\gamma$
and $\beta$ which can be point of interest for theoretical properties.

\subsection{Selection indicators}

The joint posterior of the selection indicators is given by:
\begin{align}
    P(\gamma\mid y)
    &\overset{\gamma}{\propto}\int\int\int P(\beta, \gamma, \sigma^2, q \mid y)d\beta dq d\sigma^2\\
    &\overset{\gamma}{\propto} \int\left[\int P(y\mid \beta, \sigma^2)P(\beta \mid \gamma, \sigma^2)d\beta\right]
    \left[\int P(\gamma\mid q) P(q) 
    dq\right]P(\sigma^2) d \sigma^2.
\end{align}
To simplify the expression, we first prove the following lemma:
\begin{lemma}\label{lem:int:beta}
Let $D_{\gamma}\coloneqq \text{diag}(\tau_{1}^{-2\gamma_j}\tau_0^{2(1-\gamma_j)})$,
$L_{\gamma}\coloneqq (\mathbf{x}^T\mathbf{x} + D_{\gamma})^{-1}$ and $\mu_{\gamma}\coloneqq L_{\gamma}\mathbf{x}^Ty$. Then,
\begin{equation}
    \int P(y\mid \beta,\sigma^2)P(\beta\mid \gamma,\sigma^2)d\beta = 
    \frac{1}{\sqrt{(2\pi\sigma^2)^{n}}}
    \left(\frac{\sqrt{|L_{\gamma}|}}{\tau_1^{\sum\gamma_j} \tau_0^{(p-\sum\gamma_j)}}\right)
    \exp\left(-\frac{y^Ty-{\mu_{\gamma}}^T L_{\gamma}^{-1}\mu_{\gamma}}{2\sigma^2}\right)
\end{equation}
\end{lemma}

\begin{proof} The lemma can be proven using the arithmetic properties of multivariate normal
distributions and therefore we omit the proof.
\end{proof}

Now, by using \cref{lem:int:beta} and standard results of the beta-Bernoulli model, we get
\begin{align}
    &P(\gamma\mid y)\nonumber\\
    &\overset{\gamma}{\propto}\int \tfrac{1}{\sqrt{(2\pi\sigma^2)^{n}}}
    \left(\tfrac{\sqrt{| L_{\gamma}|}}{\tau_1^{\sum\gamma_j} \tau_0^{(p-\sum\gamma_j)}}\right)
    \exp\left(-\tfrac{y^Ty-{\mu_{\gamma}}^T L_{\gamma}^{-1}\mu_{\gamma}}{2\sigma^2}\right)
    \left(\prod_{j} \alpha_j^{\gamma_j}
    (1-\alpha_j)^{1-\gamma_j}\right) P(\sigma^2) d\sigma^2\\
    &\overset{\gamma}{\propto} \left(\prod_{j} \alpha_j^{\gamma_j}
    (1-\alpha_j)^{1-\gamma_j}\right)
    \left(\tfrac{\sqrt{| L_{\gamma}|}}{\tau_1^{\sum\gamma_j} \tau_0^{(p-\sum\gamma_j)}}\right) 
    \int \tfrac{1}{\sqrt{(2\pi\sigma^2)^{n}}}
    \exp\left(-\tfrac{y^Ty-{\mu_{\gamma}}^T L_{\gamma}^{-1}\mu_{\gamma}}{2\sigma^2}\right) P(\sigma^2) d\sigma^2\\
    &\overset{\gamma}{\propto} \left(\prod_{j} \alpha_j^{\gamma_j}
    (1-\alpha_j)^{1-\gamma_j}\right)
    \left(\tfrac{\sqrt{| L_{\gamma}|}}{\tau_1^{\sum\gamma_j} \tau_0^{(p-\sum\gamma_j)}}\right)
    \\*\nonumber&\qquad\qquad\int \tfrac{1}{\sigma^{n}}
    \exp\left(-\tfrac{y^Ty-{\mu_{\gamma}}^T L_{\gamma}^{-1}\mu_{\gamma}}{2\sigma^2}\right) 
    \tfrac{1}{\sigma^{2(a+1)}}
    \exp\left(-\tfrac{b}{\sigma^2}\right) d\sigma^2\\
    &\overset{\gamma}{\propto} \left(\prod_{j} \alpha_j^{\gamma_j}
    (1-\alpha_j)^{1-\gamma_j}\right)
    \left(\tfrac{\sqrt{| L_{\gamma}|}}{\tau_1^{\sum\gamma_j} \tau_0^{(p-\sum\gamma_j)}}\right)
    \\*\nonumber&\qquad\qquad\int \tfrac{1}{\sigma^{2(n/2 +a+1)}}
    \exp\left(-\tfrac{1}{\sigma^2}
    \left(b + \tfrac{y^Ty - {\mu_{\gamma}}^T L_{\gamma}^{-1} \mu_{\gamma}}{2}\right)\right) d\sigma^2\\
    &\overset{\gamma}{\propto} \left(\prod_{j} \alpha_j^{\gamma_j}
    (1-\alpha_j)^{1-\gamma_j}\right)
    \left(\tfrac{\sqrt{| L_{\gamma}|}}{\tau_1^{\sum\gamma_j} \tau_0^{(p-\sum\gamma_j)}}\right) 
    \tfrac{\Gamma(n/2+a)}{\left(b + \tfrac{y^Ty - {\mu_{\gamma}}^T L_{\gamma}^{-1} \mu_{\gamma}}{2}\right)^{n/2+a}}.
\end{align}
Therefore,
\begin{equation}
    P(\gamma\mid y) = \frac{\left(\prod_{j} \alpha_j^{\gamma_j}
    (1-\alpha_j)^{1-\gamma_j}\right)\left(\frac{\sqrt{| L_{\gamma}|}}
    {\tau_1^{\sum\gamma_j} \tau_0^{(p-\sum\gamma_j)}}\right) 
    \frac{1}{\left(b + \frac{y^Ty - {\mu_{\gamma}}^T L_{\gamma}^{-1} \mu_{\gamma}}{2}\right)^{n/2+a}}}
    {\sum_{\gamma}\left(\prod_{j} \alpha_j^{\gamma_j}
    (1-\alpha_j)^{1-\gamma_j}\right)\left(\frac{\sqrt{| L_{\gamma}|}}
    {\tau_1^{\sum\gamma_j} \tau_0^{(p-\sum\gamma_j)}}\right) 
    \frac{1}{\left(b + \frac{y^Ty - {\mu_{\gamma}}^T L_{\gamma}^{-1} \mu_{\gamma}}{2}\right)^{n/2+a}}}.\label{eq:joint:post:gamma}
\end{equation}

Clearly, we get the most probable model when $P(\gamma\mid y)$ is maximal. However, for that we need to 
search a space of size $2^p$ and therefore it is not feasible to do that for very large 
values of $p$, which is often the case for high dimensional problems.

\Cref{eq:joint:post:gamma} also allow us to obtain the posterior odds between
two different models. Let $\gamma^{'}$ and $\gamma^{''}$ denote two
different models. Then, we have
\begin{align}
    \frac{P(\gamma = \gamma^{'}\mid y)}{P(\gamma= \gamma^{''}\mid y)}
    & = \frac{\left(\prod_{j} \alpha_j^{\gamma^{'}_j}
    (1-\alpha_j)^{1-\gamma^{'}_j}\right)\left(\frac{\sqrt{| L_{\gamma^{'}}|}}
    {\tau_1^{\sum\gamma^{'}_j} \tau_0^{(p-\sum\gamma^{'}_j)}}\right) 
    \frac{1}{\left(b + \frac{y^Ty - {\mu_{\gamma^{'}}}^T L_{\gamma^{'}}^{-1} \mu_{\gamma^{'}}}{2}\right)^{n/2+a}}}
    {\left(\prod_{j} \alpha_j^{\gamma^{''}_j}
    (1-\alpha_j)^{1-\gamma^{''}_j}\right)\left(\frac{\sqrt{| L_{\gamma^{''}}|}}
    {\tau_1^{\sum\gamma^{''}_j} \tau_0^{(p-\sum\gamma^{''}_j)}}\right) 
    \frac{1}{\left(b + \frac{y^Ty - {\mu_{\gamma^{''}}}^T L_{\gamma^{''}}^{-1} \mu_{\gamma^{''}}}{2}\right)^{n/2+a}}}\\
    & = \frac{\left(\prod_{j} 
    \left(\frac{1-\alpha_j}{\alpha_j}\right)^{1-\gamma^{'}_j}\right)
    \left(\frac{\tau_0}{\tau_1}\right)^{\sum\gamma^{'}_j}
    \sqrt{| L_{\gamma^{'}}|} 
    \frac{1}{\left(b + \frac{y^Ty - {\mu_{\gamma^{'}}}^T L_{\gamma^{'}}^{-1} \mu_{\gamma^{'}}}{2}\right)^{n/2+a}}}
    {\left(\prod_{j} \left(\frac{1-\alpha_j}{\alpha_j}\right)^{1-\gamma^{''}_j}\right)
    \left(\frac{\tau_0}{\tau_1}\right)^{\sum\gamma^{''}_j}
    \sqrt{| L_{\gamma^{''}}|} 
    \frac{1}{\left(b + \frac{y^Ty - {\mu_{\gamma^{''}}}^T L_{\gamma^{''}}^{-1} \mu_{\gamma^{''}}}{2}\right)^{n/2+a}}}\\
    & = \left(\prod_{j} 
    \left(\frac{1-\alpha_j}{\alpha_j}\right)^{\gamma^{''}_j-\gamma^{'}_j}\right)
    \left(\frac{\tau_0}{\tau_1}\right)^{\sum(\gamma^{'}_j - \gamma^{''}_j)}
    \frac{\sqrt{| L_{\gamma^{'}}|}}{\sqrt{| L_{\gamma^{''}}|}}
    \left(\frac{b + \frac{y^Ty - {\mu_{\gamma^{''}}}^T L_{\gamma^{''}}^{-1} \mu_{\gamma^{''}}}{2}}
    {b + \frac{y^Ty - {\mu_{\gamma^{'}}}^T L_{\gamma^{'}}^{-1} \mu_{\gamma^{'}}}{2}}\right)^{n/2+a}.\label{eq:joint:gamma:post:odds}
\end{align}
We can see that the posterior odds between two models are monotone with
respect to the prior expectations of the selection probabilities.  This 
suggests that incorrect values of $\alpha$ will affect the 
posterior model probabilities and we may not get the most probable model. 
Therefore, it is beneficial to have a cautious 
approach when we are unsure of the true values of $\alpha$
The monotonicity also suggests that we need to compute the bounds of the posterior odds on the
extreme values of $\alpha$ similar to our result for orthogonal design
case. However, as mentioned earlier, this will
require $2\cdot2^p$ number of computations to compute all such posterior
odds, which is not very practical. Instead,
we will be employing a Gibbs sampling scheme to compute the marginals of $\gamma$
to evaluate the posterior odds of $\gamma_j$.

}

\addition{
\subsection{Regression coefficients}
Similar to the joint posterior of the selection indicators, we can also compute the joint posterior
of the regression coefficients.

Let $r_{\gamma} = \frac{y^Ty - y^T\mathbf{x}L_{\gamma}\mathbf{x}^Ty}{2} +b, \text{ and }
    \Sigma^{-1}_{\gamma} = \frac{n+2a}{2r_{\gamma}}L_{\gamma}^{-1}$.
Then the joint posterior of $\beta$ is given
by:
\begin{align}
    P(\beta \mid y)
    &=\frac{\sum_{\gamma}\left(
    \left(\prod_{j} \alpha_j^{\gamma_j}
    (1-\alpha_j)^{1-\gamma_j}\right)\left(\frac{\sqrt{| L_{\gamma}|}}
    {\tau_1^{\sum\gamma_j} \tau_0^{(p-\sum\gamma_j)}}\right) 
    \frac{1}{\left(b + \frac{y^Ty - {\mu_{\gamma}}^T L_{\gamma}^{-1} \mu_{\gamma}}{2}\right)^{n/2+a}}
    \mathcal{T}_{n+2a}(\mu_{\gamma}, \Sigma_{\gamma})
    \right)}
    {\sum_{\gamma}\left(
    \left(\prod_{j} \alpha_j^{\gamma_j}
    (1-\alpha_j)^{1-\gamma_j}\right)\left(\frac{\sqrt{| L_{\gamma}|}}
    {\tau_1^{\sum\gamma_j} \tau_0^{(p-\sum\gamma_j)}}\right) 
    \frac{1}{\left(b + \frac{y^Ty - {\mu_{\gamma}}^T L_{\gamma}^{-1} \mu_{\gamma}}{2}\right)^{n/2+a}}\right)},
\end{align}
where $\mathcal{T}_{n+2a}(\mu, \Sigma)$ denotes a multivariate t-distribution with
mean $\mu$, scale-matrix $\Sigma$ and degrees of freedom $n+2a$. 
The calculation of the joint posterior requires some non-trivial arithmetic manipulations,
which is too lengthy for the main text. Please check appendix (\cref{app:sec:rc})
for more details.

Now, as we see from the above expression, the joint posterior
of $\beta$ can be represented as a weighted mixture of multivariate t-distributions,
where the weight of each component corresponds to the posterior
model selection probability. This is very similar to our result for the orthogonal design case, where instead of joint posterior of $\beta$, we had similar expression
for $\beta_j$ for each $j$. 

The closed form expression of the posterior of $\beta$ allows us to
obtain the posterior expectation of $\beta$, which is given by:

\begin{equation}
    E(\beta \mid y) = 
    \frac{\sum_{\gamma}\left(
    \left(\prod_{j} \alpha_j^{\gamma_j}
    (1-\alpha_j)^{1-\gamma_j}\right)\left(\frac{\sqrt{| L_{\gamma}|}}
    {\tau_1^{\sum\gamma_j} \tau_0^{(p-\sum\gamma_j)}}\right) 
    \frac{1}{\left(b + \frac{y^Ty - {\mu_{\gamma}}^T L_{\gamma}^{-1} \mu_{\gamma}}{2}\right)^{n/2+a}}
    \mu_{\gamma}
    \right)}
    {\sum_{\gamma}\left(
    \left(\prod_{j} \alpha_j^{\gamma_j}
    (1-\alpha_j)^{1-\gamma_j}\right)\left(\frac{\sqrt{| L_{\gamma}|}}
    {\tau_1^{\sum\gamma_j} \tau_0^{(p-\sum\gamma_j)}}\right) 
    \frac{1}{\left(b + \frac{y^Ty - {\mu_{\gamma}}^T L_{\gamma}^{-1} \mu_{\gamma}}{2}\right)^{n/2+a}}\right)}.
\end{equation}

Unlike before, we do not have an expression to show the role of the posterior in variable selection
as we do not have least square estimates to start with. However, one may try to obtain
such expression using Ridge estimates. This may lead to a contour where we can define
a region of indeterminacy similar to what we showed in orthogonal cases.

}
\subsection{Gibbs sampling algorithm}
We provided the joint posteriors of the selection indicators and the
regression coefficients. These are not particularly useful for parameter
estimation.
Therefore, we need a suitable computation scheme for general cases. 
Interestingly, our choice of priors allows us to obtain closed expressions for the
full conditional distributions of the modelling parameters. From these expressions, we can easily implement a Gibbs sampling routine \cite{gelfand_gibbs}
to compute posterior distributions and hence perform variable selection. In the rest of this section, we derive these expressions.

Recall the joint posterior in \cref{eq:joint:post}. Then 
the joint conditional distribution of the regression coefficients is given by:
\begin{align}
    P(\beta\mid \gamma, \sigma^2, q, y)
    &\overset{\beta}{\propto} P(\beta, \gamma, \sigma^2, q \mid y)\\
    &\overset{\beta}{\propto} P(y\mid \beta, \sigma^2)P(\beta \mid \gamma, \sigma^2)\\
    &\overset{\beta}{\propto} \exp\left(-\frac{1}{2\sigma^2}\|y-\mathbf{x}\beta\|^2_2\right)
    \prod_{j=1}^{p} f_{\gamma_j}(\beta_j)\\
    &\overset{\beta}{\propto} \exp\left(-\frac{\beta^T\mathbf{x}^T\mathbf{x}\beta - 2\beta^T\mathbf{x}^Ty}
    {2\sigma^2}\right)\prod_{j=1}^{p} f_{\gamma_j}(\beta_j)\\
    &\overset{\beta}{\propto} \exp\left(-\frac{\beta^T\mathbf{x}^T\mathbf{x}\beta - 2\beta^T\mathbf{x}^Ty}
    {2\sigma^2}\right)\prod_{j=1}^{p}
    \exp\left(-\frac{\beta_j^2}{2\sigma^2\tau_{\gamma_j}^2}\right)\label{eq:full:cond:beta:1}.
\end{align}
Let $D_{\gamma}\coloneqq \text{diag}(\tau_{1}^{-2\gamma_j}\tau_0^{2(1-\gamma_j)})$, then we can rewrite
\cref{eq:full:cond:beta:1} as
\begin{align}
    P(\beta\mid \gamma, \sigma^2, q, y)
    &\overset{\beta}{\propto} \exp\left(-\frac{\beta^T\mathbf{x}^T\mathbf{x}\beta - 2\beta^T\mathbf{x}^Ty}
    {2\sigma^2}\right)\exp\left(-\frac{\beta^T D_{\gamma}\beta}{2\sigma^2}\right)\\
    &\overset{\beta}{\propto}\exp\left(-\frac{\beta^T\mathbf{x}^T\mathbf{x}\beta - 2\beta^T\mathbf{x}^Ty
    +\beta^T D_{\gamma}\beta}{2\sigma^2}\right)\\
    &\overset{\beta}{\propto} \exp\left(-
    \frac{(\beta-\mu_{\gamma})^TL_{\gamma}^{-1}(\beta-\mu_{\gamma})}
    {2\sigma^2}\right)
\end{align}
where, \addition{$L_{\gamma}\coloneqq (\mathbf{x}^T\mathbf{x} + D_{\gamma})^{-1}$ 
and $\mu_{\gamma}\coloneqq L_{\gamma}\mathbf{x}^Ty$}. Therefore
the full conditional of $\beta$ follows a normal distribution such that:
\begin{equation}
    \beta\mid \gamma, \sigma^2, q, y \sim \mathcal{N}(\mu_{\gamma}, \sigma^2L_{\gamma}).
\end{equation}

For the selection indicators , we only need to compute the probability of $\gamma_j$
conditional on $\beta$, $\sigma^2$ and $q_j$. Therefore, we can compute these posteriors
component wise:
\begin{align}
    P(\gamma_j \mid \beta_j, \sigma^2, q_j)
    &\overset{\gamma_j}{\propto} P(\beta_j\mid \gamma_j, \sigma^2) P(\gamma_j\mid q_j)\\
    &\overset{\gamma_j}{\propto} q_j^{\gamma_j}(1-q_j)^{1-\gamma_j} f_{\gamma_j}(\beta_j)\\
    &\overset{\gamma_j}{\propto} [q_jf_{\gamma_j}(\beta_j)]^{\gamma_j}
    [(1-q_j)f_{\gamma_j}(\beta_j)]^{1-\gamma_j}.
\end{align}
Therefore, $\gamma_j \mid \beta_j, \sigma^2$ follows a Bernoulli distribution with
\begin{equation}
    P(\gamma_j = 1 \mid \beta_j, \sigma^2) = 
    \frac{q_jf_{1}(\beta_j)}{q_jf_{1}(\beta_j)+(1-q_j)f_{0}(\beta_j)}.
\end{equation}

Unlike orthogonal design case, the choice of concentration parameter plays an important
role on the conditional distributions of $q_j$'s for the Gibbs sampling algorithm. The conditional
distribution of $q_j$'s follows a beta distribution
\begin{equation}
    q_j\mid \gamma_j \sim \text{Beta} (s\alpha_j + \gamma_j, s(1 - \alpha_j)+ 1 - \gamma_j),
\end{equation}
where $\alpha_j\in\mathcal{P}$. 

The conditional distribution of $\sigma^2$ is given by:
\begin{align}
    &P(\sigma^2\mid \beta, \gamma, y)\nonumber\\
    &\overset{\sigma^2}{\propto} P(y\mid \beta, \sigma^2)P(\beta\mid \gamma, \sigma^2)P(\sigma^2)\\
    &\overset{\sigma^2}{\propto}\frac{1}{\sigma^n}\exp\left(-\frac{\|y-\mathbf{x}\beta\|_2^2}
    {2\sigma^2}\right)\frac{1}{\sigma^p}\exp\left(-\frac{\beta^T D_{\gamma}\beta}{2\sigma^2}\right)\frac{1}{\sigma^{2(a+1)}}
    \exp\left(-\frac{b}{\sigma^2}\right)\\
    &\overset{\sigma^2}{\propto}\frac{1}{\sigma^{2(p/2+ n/2 + a+1)}}
    \exp\left\{-\frac{1}{\sigma^2}\left(
    \frac{\|y-\mathbf{x}\beta\|_2^2}{2} + \frac{\beta^T D_{\gamma}\beta}{2} +b\right)\right\}
\end{align}
Therefore, 
\begin{equation}
    \sigma^2\mid \beta, \gamma, y \sim 
    \text{IG}\left(a + \frac{p}{2} + \frac{n}{2}, b + 
    \frac{\|y-\mathbf{x}\beta\|_2^2}{2} + \frac{\beta^T D_{\gamma}\beta}{2}\right).
\end{equation}

\section{Illustration}\label{sec:illustrate}

In this section,  
\addition{we demonstrate our sensitivity analysis based 
variable selection approach for both synthetic and real datasets.
For that, we will define different accuracy measures,
which will be used for model fitting.

\subsection{Accuracy Measures}
So far, we discussed how we wish to perform a sensitivity analysis
over a set of priors to obtain a cautious variable selection scheme.
In this section, we discuss different aspects of model fitting 
in our sensitivity analysis based paradigm. As mentioned earlier, our method provides a 
set of posteriors instead of a single posterior. Therefore, during model
fitting, we will also have multiple models instead of a single model.
Therefore, to discuss model fitting, we need 
different measures. To define those measures, we first consider the following
set
\begin{equation}
    \mathcal{A}(\alpha)\coloneqq \left\{j:\left\{\frac{P\left(\gamma_j = 1\mid y\right)}
    {P\left(\gamma_j = 0\mid y\right)}\right\} > 1\right\}.
\end{equation}
Therefore, $\mathcal{A}(\alpha)$ denotes the set of active 
variables for each value of $\alpha$. 

Now, we define
minimum squared error so that
\begin{equation}
    \text{Minimum Squared Error} = 
    \min_{\alpha\in \mathcal{P}}\left\|Y-
    \mathbf{X}_{\mathcal{A}(\alpha)}\hat{\beta}^{\text{post}}_{\mathcal{A}(\alpha)}
    \right\|_2^2
\end{equation}
and maximum squared error so that
\begin{equation}
    \text{Maximum Squared Error} = 
    \max_{\alpha\in \mathcal{P}}\left\|Y-
    \mathbf{X}_{\mathcal{A}(\alpha)}\hat{\beta}^{\text{post}}_{\mathcal{A}(\alpha)}
    \right\|_2^2
\end{equation}
where 
$\hat{\beta}^{\text{post}}_{\mathcal{A}(\alpha)}\coloneqq E(\beta_{\mathcal{A}(\alpha)}\mid Y)$
is the posterior mean of $\beta$ for the selected subset 
$\mathcal{A}(\alpha)$. This way, we can get an optimistic fit and a pessimistic
fit by using the minimum squared error and maximum squared error respectively. 
Clearly, for other methods we do not have a dependence on $\alpha$. So,
to compare with other methods, we consider the conventional squared error
given by:
\begin{equation}
    \text{Squared Error} = 
    \left\|Y-\mathbf{X}_{\mathcal{A}}\hat{\beta}^{\text{post}}_{\mathcal{A}}
    \right\|_2^2,
\end{equation}
where, $\mathcal{A}$ is the set of selected variables and 
$\hat{\beta}^{\text{post}}_{\mathcal{A}}$ is the posterior estimates
of the regression coefficients of the selected variables.

Having multiple models also creates an indeterminacy in prediction. 
Therefore, we can also capture the model indeterminacy
using the following measure:
\begin{equation}
    \text{Model Indeterminacy} = \frac{\text{Maximum Squared Error} - \text{Minimum Squared Error}}
    {\text{Maximum Squared Error}}.
\end{equation}
Clearly, model indeterminacy can be non-zero, even when there are zero indeterminate variables present in the model.
}

\subsection{Synthetic Datasets}
In this section we will show the accuracy of our method in terms of variable 
selection \addition{as well as model fitting}. We construct four 
different synthetic datasets to investigate different aspects
of variable selection problems. We generate the 
design matrix ($\mathbf{x}$) with 100 predictors ($p$) and 50 observations ($n$) so that
$\mathbf{x}_i\sim \mathcal{N}(0, \Sigma)$ for $1\le i \le 50$, \addition{
where $[\Sigma]_{ij} = 0.2 ^{|i-j|}$.}
We consider four different sets of regression coefficients to generate the response under the
model defined in \cref{eq:lm}.
\begin{itemize}
    \item Dataset 1: 10 active predictors (sparse model)
    \item Dataset 2: 20 active predictors (fairly sparse model)
    \item Dataset 3: 50 active predictors (fairly dense model)
    \item Dataset 4: 60 active predictors (dense model)
\end{itemize}
In all of the above cases, we generate the true regression coefficients $\beta^*$
from $U([-4,-1]\cup [1,4])$ and assign a random noise \addition{$\epsilon_i\sim\mathcal{N}(0,4)$}.

\addition{To evaluate the accuracy in} model fitting, \addition{we simply
use the minimum and maximum squared
errors. But, synthetic datasets also allow us to check the accuracy of 
parameter estimation, and in order to evaluate that, we need} to compare the posterior estimates
of the regression coefficients with \addition{their} true values. 
\addition{Let $\hat{\beta}^{\text{post}}\coloneqq E(\beta\mid Y)$.
That is, $\hat{\beta}^{\text{post}}$ denotes the posterior
expectation of $\beta$. Then we define
the following:
\begin{equation}
    \Delta(\beta)\coloneqq 
    \sum_{j=1}^p (\hat{\beta}^{\text{post}}_{j}\cdot 
    \mathbb{I}_{j\in \mathcal{A}} -\beta^*_j)^2,
\end{equation}
where $\mathbb{I}_{j\in \mathcal{A}}$ is the indicator that $j$-th
variable is in $\mathcal{A}$. For our sensitivity analysis based
approach, this $\Delta(\beta)$ can be seen as a function of $\alpha$.}

\paragraph{Results}
\addition{For illustration}, we consider two different
sets for $\alpha_j$ so that one set represents a near vacuous case and the other set 
represents elicited prior information. To specify the near vacuous case, we consider 
$\alpha_j\in[0.05,0.95]$ for the $j$-th co-variate. The choice of the
elicitation based sets is dependent on the examples. \addition{We fit a ridge regression
model on the data and then calculate the $p$-values of the regression estimates
to understand the total number of active co-variates in the model.}
For instance, for the first synthetic dataset, \addition{we see that we can have 5 to 12 different
active co-variates based on different tolerances on the $p$-value. Therefore,} we set 
\addition{$\alpha_j\in [0.05, 0.12]$}. Similarly, we consider 
\addition{$\alpha_j\in[0.08,0.22]$} for the second, 
\addition{$\alpha_j\in[0.10,0.33]$} for the third and 
\addition{$\alpha_j\in[0.16,0.34]$} for the fourth synthetic dataset. 
We fix $\tau_1 = 5$ for the near vacuous case  and $\tau_1=1$ for the elicitation based case 
and for all the cases, we fix $\tau_0 = 10^{-6}$
to ensure a spike at 0.

\addition{
\begin{table}[ht]
    \centering
    \tiny
    \renewcommand{\arraystretch}{.95}
    \begin{tabular}{|l|ccccrc|}
  \hline
 Methods & Act & FA & Inact & FI &
 Sq.~E 
 &  $\Delta(\beta)$\\ 
      \hline
      \multicolumn{7}{|c|}{Synthetic dataset 1, active 10 and inactive 90}\\
      \hline
CBVS-NV (optimistic) & 8-1 & 0-0 & 0-91 & 0-1 & 61.04 & 2.39 \\ 
CBVS-NV (pessimistic) & 8-92 & 0-90 & 0-0 & 0-0 & 470.81 & 36.63 \\
  CBVS-E (optimistic) & 9-0 & 0-0 & 90-1 & 0-1 & 64.54 & 1.86 \\ 
  CBVS-E (pessimistic) & 9-0 & 0-0 & 90-1 & 0-1 & 69.81 & \textbf{1.82} \\ 
  Horseshoe & 8 & 0 & 92 & 2 & 77.86 & 5.71 \\ 
  Spike and Slab & 8 & 0 & 92 & 2 & 212.29 & 16.92 \\ 
  Bayesian LASSO & 13 & 3 & 87 & 0 & 69.06 & 2.23 \\ 
  MCP & 13 & 3 & 87 & 0 & \textbf{56.77} & 2.61 \\  
  SCAD & 15 & 5 & 85 & 0 & 56.94 & 2.60 \\ 
  LASSO & 22 & 12 & 78 & 0 & 62.53 & 2.01 \\
  \hline
  \multicolumn{7}{|c|}{Synthetic dataset 2, active 20 and inactive 80}\\
  \hline
  CBVS-NV (optimistic) & 2-28 & 0-16 & 0-70 & 0-6 & 457.13 & 78.63 \\ 
  CBVS-NV (pessimistic) & 2-31 & 0-24 & 0-67 & 0-11 & 2962.91 & 238.98 \\
  CBVS-E (optimistic) & 8-8 & 1-0 & 79-5 & 4-1 & \textbf{188.33} & \textbf{23.45} \\ 
  CBVS-E (pessimistic) & 8-5 & 1-4 & 79-8 & 4-8 & 1533.69 & 94.93 \\
  Horseshoe & 8 & 1 & 92 & 13 & 808.71 & 73.31 \\ 
  Spike and Slab & 5 & 0 & 95 & 15 & 1521.08 & 98.29 \\ 
  Bayesian LASSO & 26 & 10 & 74 & 4 & 676.22 & 48.77 \\ 
  MCP & 20 & 12 & 80 & 12 & 1615.63 & 105.10 \\ 
  SCAD & 27 & 17 & 73 & 10 & 1650.48 & 102.96 \\ 
  LASSO & 42 & 25 & 58 & 3 & 602.24 & 43.13 \\
  \hline
      \multicolumn{7}{|c|}{Synthetic dataset 3, active 50 and inactive 50}\\
      \hline
  CBVS-NV (optimistic) & 2-44 & 1-18 & 0-54 & 0-23 & \textbf{1181.73} & 376.48 \\ 
  CBVS-NV (pessimistic) & 2-2 & 1-0 & 0-96 & 0-47 & 4940.32 & 476.50 \\ 
  CBVS-E (optimistic) & 8-9 & 4-0 & 77-6 & 34-3 & 1700.08 & 337.18 \\ 
  CBVS-E (pessimistic) & 8-1 & 4-1 & 77-14 & 34-12 & 4028.74 & 433.16 \\
  Horseshoe & 6 & 2 & 94 & 46 & 2130.42 & 392.36 \\
  Spike and Slab & 9 & 2 & 91 & 43 & 4143.12 & 375.22 \\ 
  Bayesian LASSO & 22 & 7 & 78 & 35 & 2095.79 & \textbf{310.55} \\ 
  MCP & 12 & 4 & 88 & 42 & 5206.86 & 477.08 \\ 
  SCAD & 14 & 5 & 86 & 41 & 4861.73 & 474.14 \\
  LASSO & 47 & 18 & 53 & 21 & 1826.17 & 337.08 \\
  \hline
      \multicolumn{7}{|c|}{Synthetic dataset 4, active 60 and inactive 40}\\
      \hline
  CBVS-NV (optimistic) & 1-28 & 0-5 & 0-71 & 0-36 & \textbf{1356.60} & 548.39 \\
  CBVS-NV (pessimistic) & 1-32 & 0-11 & 0-67 & 0-38 & 10552.30 & 837.40 \\ 
  CBVS-E (optimistic) & 10-16 & 0-4 & 70-4 & 36-2 & 2151.60 & 294.74 \\ 
  CBVS-E (pessimistic) & 10-3 & 0-2 & 70-17 & 36-13 & 4172.37 & 420.02 \\ 
  Horseshoe & 2 & 0 & 98 & 58 & 3702.57 & 387.81 \\
  Spike and Slab & 11 & 1 & 89 & 50 & 3768.20 & 383.40 \\ 
  Bayesian LASSO & 23 & 4 & 77 & 41 & 3078.20 & 300.71 \\
  MCP & 6 & 0 & 94 & 54 & 4417.04 & 394.47 \\ 
  SCAD & 16 & 1 & 84 & 45 & 3316.51 & 353.27 \\  
  LASSO & 40 & 10 & 60 & 30 & 2464.52 & \textbf{276.88} \\ 
  \hline
    \end{tabular}
\caption{Summary of variable selection and model fitting for 4 different synthetic datasets. 
The first four rows in each section show the results of our robust
Bayesian analysis and therefore the numbers are separated by hyphens as some 
of the variables are indeterminate in our analyses but are active (inactive)
in the concerning model.}
\label{tab:summary:vs}
\end{table}
}

We provide the result of our analyses in \cref{tab:summary:vs}. \addition{In the table,
we denote our cautious method by CBVS, followed by two different abbreviations;
NV for near vacuous case and E for the elicitation based case. We also} 
use six other methods for comparison, which are horseshoe (half-Cauchy), spike and slab prior, Bayesian LASSO,
MCP (minimax concave penalty \cite{zhang2010}), SCAD and LASSO. These methods are performed using the R packages 
\texttt{horseshoe}, \texttt{spikeslab}, \texttt{monomvn} (for Bayesian LASSO), \texttt{ncvreg} (for MCP and SCAD) and
\texttt{glmnet}. 

In the first column of \addition{\cref{tab:summary:vs}}, we show the number of active variables
present in the model.
\addition{For our cautious variable selection method, two numbers are provided separated by a hyphen.
The left side of the hyphen shows the number of active variables,
identified by the
method and the right side shows the number of indeterminate variables (obtained from 
our robust Bayesian analysis), which are considered to be  active
in the model. For instance, for our near vacuous
analysis with the first dataset, 
we get 1 indeterminate variable to be active in the optimistic fit
but 92 indeterminate variables appears to be active in the pessimistic
model}. In the next column, we present
the \addition{number of falsely active variables. Similar to the first column, we also
split this number in two for our robust Bayesian analyses}. The third 
\addition{and the fourth} columns show \addition{the results for inactive 
variables. Clearly, as other methods do not produce any indeterminate
variable, there are only single numbers for these methods.}
In the next two columns, we \addition{provide the squared errors and the}  
deviation of the posterior estimates of the regression coefficients with 
respect to their true value.

For the first \addition{dataset, we notice that our method gives us 8 active variables
and 92 indeterminate variables for the near vacuous case. This gives us
decent result for the optimistic fit which is not the case for
the pessimistic fit. We see that all of the indeterminate
variables are considered as active and as a result we get really high values
for the squared error and $\Delta(\beta)$. This may seem obvious and 
suggests us to ignore the pessimistic fits in general but that might not
be the right thing to do. We can see this from the elicitation 
based analysis. In this case, we see that the optimistic fit and pessimistic
model have very similar outcomes in terms of variable selection and
model fitting. However, in this particular case the pessimistic fit
performs better than the optimistic fit in terms
of parameter estimation and gives us lower values of $\Delta(\beta)$.
We also notice that the frequentist methods tend to do better for this
dataset, however these methods also tend to select more variables in the
model.

For the second dataset, we notice that our elicitation based case performs
really well. The optimistic fit of CBVS-E outperforms every other
model in terms of model fitting as well as parameter estimation. It
is also fairly accurate in terms of variable selection. The elicitation
based case identifies 8 active variables and 13 indeterminate variables.
In the optimistic fit, it considers 8 of these 13 indeterminate
variables, which leads to a better performance in variable selection.
Some other methods such as Bayesian LASSO and LASSO includes more
number of active variables in the model. However, both of these
methods also tend to include too many false active variables in the
model which increases the value of squared error and $\Delta(\beta)$.

The results for third and fourth datasets show the importance of performing
analysis with a near vacuous set of prior distribution. Both of these datasets are dense, 
that is the number of true active variables is very high. In general, it  
is very difficult to elicit information from high dimensional dense models.
We can notice that from the initial
analysis with ridge estimates. In both cases, ridge estimates suggest
us fewer numbers of active variables in the datasets than there should
be. We do not face this problem with a
near vacuous set of priors and we get better results from the optimistic fits
of the near vacuous analyses. This happens as the near vacuous 
case allows us to include more variables in the model, which is important
for dense datasets. However, the near vacuous case may give us very  poor
results for the pessimistic fit, which is evident from our analysis
using the fourth dataset. This leads to a very high indeterminacy
of $0.87$ as well, which is not desirable and suggests that we must gather
more information. We also notice that
the elicitation based analysis performs much better in parameter estimation.
This happens as the total number of misspecified 
variables is lower for the elicitation based case. This is also the
case for Bayesian LASSO, which tends to overshrink the regression
coefficients to attain sparsity. This results to fewer 
false positives and keeps the value of
$\Delta(\beta)$ in control.

}

\subsection{Real Data Analysis}
We use three real datasets to inspect different aspects of high dimensional problems.
For each dataset, two different sets of $\alpha$ are considered to represent both
near vacuous case and elicitation based case.

\paragraph{Diabetes dataset}
The Diabetes dataset \cite{efron2004} features 
10 predictors which are age, sex, body mass index, average blood pressure and six blood serum
measurements. The response denotes the disease progression in one year

Preliminary analysis such as ordinary least squares suggests that this dataset 
contains 2 to 5 active variables depending upon our choice. 
Based on this preliminary analysis,
we consider two different sets to specify our prior expectation of the selection 
probabilities denoted by $\alpha\coloneqq(\alpha_1$, \dots, $\alpha_p)$. We first 
specify a near vacuous set so that, $\alpha_j\in [0.1, 0.9]$ and we choose the other
set so that $\alpha_j\in[0.2, 0.5]$. Therefore, our second choice of $\alpha_j$'s is
a direct representation of our prior information on the selection probability of 
variables.

\paragraph{Gaia dataset}
The Gaia dataset was
used for computer experiments \cite{jones2010,Einbeck2008} prior to the launch of European
Space Agency's Gaia mission. 
The data contains 8286 
observations on the spectral information of 16 $(p)$
wavelength bands, and four different stellar parameters. In this example, we take 
stellar-temperature (in Kelvin scale) as the response variable. 

The variables in the Gaia dataset are highly correlated and 
previous work by \citet{Einbeck2008} suggests that there are only 1-3 main 
contributory variables. Based on this information, we take two sets for $\alpha_j$
similar to our choice of $\alpha_j$ for Diabetes dataset.
We specify our first set to specify a near vacuous set and choose $\alpha_j\in[0.1,0.9]$.
The second set is based on our prior information on the contributory variables and we set
$\alpha_j\in [0.0625, 0.1875]$. 

\paragraph{Lymphoma dataset}
We investigate the Lymphoma dataset \cite{Alizadeh2000} to illustrate our result for 
a high dimensional problem. In this dataset, there are 7399 genes related to B-cell
Lymphoma along with the response which denotes censored survival times. There are
only 240 observations in this dataset which makes the problem ultra high dimensional
that is $p\gg n$. Performing Bayesian analysis in this type of dataset is 
extremely difficult and we use a variable screening method to identify 200 
important co-variates. We use the package \texttt{VariableScreening} 
to obtain the first 200 co-variates based on the correlation distance. 

The choice of selection probability for this dataset is difficult and we choose $\alpha_j$ based
on the selected co-variates after the variable screening. We fit a ridge regression
model to examine the $p$-values. This preliminary analysis suggests that we may
consider 20 to 30 variables based on our tolerance. Therefore, we specify our
elicitation based set as $\alpha_j\in[0.1,0.15]$. For the near vacuous case, we stick to
our previous examples and choose $\alpha_j\in[0.1,0.9]$.

\paragraph{Results} Similar to our analyses for synthetic datasets, we use six other methods
to compare with our robust Bayesian analysis. We show the results in \cref{tab:real:sum}.
The first \addition{two} columns (from left) in the table show the number 
of active and inactive variables. \addition{For our sensitivity based
approach (CBVS), we represent these variables so that the first part shows the
number of active (inactive) variables and the second part shows the
number of indeterminate variables found in our robust Bayesian analysis, which appear to be 
active (inactive) in the model.} \addition{These are} followed by \addition{the} squared error.

For the Diabetes dataset, \addition{we see that our optimistic fit for the elicitation
based analysis performs really well in terms of the squared error and only
Bayesian LASSO outperforms our method.
In this particular case, we see that  2 of the 
3} indeterminate variables \addition{remain active in the model, whereas
for the pessimistic fit, all 3 of them are active. For
the near vacuous based case, we have a total of 5 indeterminate variables,
out of which 3 are active for the optimistic case and all of them are active
for the pessimistic case.} We also notice that for both these cases, 
our method selects 3rd and 9th covariates to be active, similar 
to our analysis with horseshoe prior. \addition{We can see that
overall indeterminacy for both CBVS-NV and CBVS-E is very low, which
suggest that the data used here is enough to perform a standard Bayesian 
analysis. } 

\addition{The analyses with the Gaia dataset is very interesting. As mentioned 
earlier, this dataset is highly correlated and we can also notice the
effect of correlation through these analyses.} 
Our robust Bayesian analyses consider the sixth covariate to be active.
\addition{This is also the case for our analysis with the  horseshoe prior}. 
However, unlike the previous case, our analysis with the near vacuous 
set does not identify any variable as inactive \addition{and in fact, in the
optimistic fit given by CBVS-NV, all the variables are considered
to be active. This is a problem of overfitting in regression 
models using correlated data and shows the importance of prior
elicitation}. This also results \addition{in} a higher model indeterminacy 
\addition{in the near vacuous analysis}, which is not the case for 
elicitation based \addition{analysis. In the elicitation based
case, we see that the optimistic fit gives us a total of 3 active variables
out of which 2 are indeterminate variables. The performance
of the optimistic fit in terms of model fitting is also good and is indeed
better than the other Bayesian methods.}

For the Lymphoma dataset, we notice that only 3 variables remain active
in the model after our robust Bayesian analyses. This is also the 
case for the Bayesian LASSO. However, for this particular dataset, 
our analysis with horseshoe prior produces the null model and
the output remains the same after several replications. 
We also notice that for this dataset, the likelihood based methods
such as LASSO, SCAD or MCP perform better than the Bayesian methods. 
These methods tend to include more covariates in the model unlike the 
Bayesian alternatives. But, \addition{that might not be the reason 
behind the enhanced performance as the pessimistic fit in CBVS-NV
includes a total of 39 variables (36 indeterminate) as active in the
model and yet the squared error remains very high. High squared
error of the pessimistic fit also leads to 
a very high indeterminacy ($0.72$) for the near vacuous case and shows that
some elicitation needs to be done for a more reliable answer. 
This can also be verified from the table, where we can see that the 
indeterminacy becomes less for the elicitation based case.}  

\begin{table}[ht]
    \centering
    \scriptsize
    \begin{tabular}{|l|ccc|}
  \hline
Method & Active & Inactive & Sq.~Err\\ 
  \hline
    \multicolumn{4}{|c|}{Diabetes Dataset (p = 10, n = 100) }\\
  \hline
CBVS-NV (optimistic) & 2-3 & 3-2 & {6.24e+04} \\ 
  CBVS-NV (pessimistic) & 2-5 & 3-0 & 6.61e+04 \\ 
  CBVS-E (optimistic) & 2-2 & 5-1 & 6.18e+04 \\ 
  CBVS-E (pessimistic) & 2-3 & 5-0 & 6.58e+04 \\ 
  Horseshoe & 2 & 8 & 6.39e+04 \\ 
  Spike and Slab & 9 & 1 & 6.59e+04  \\ 
  Bayesian LASSO & 4 & 6 & \textbf{6.10e+04} \\ 
  MCP & 7 & 3 & 7.05e+04 \\ 
  SCAD & 8 & 2 & 7.18e+04 \\ 
  LASSO & 8 & 2 & 6.52e+04 \\  
  \hline
    \multicolumn{4}{|c|}{Gaia Dataset (p = 16, n = 100) }\\
  \hline
CBVS-NV (optimistic) & 1-15 & 0-0 & \textbf{2.74e+08} \\ 
  CBVS-NV (pessimistic) & 1-0 & 0-15 & 3.64e+08 \\ 
  CBVS-E (optimistic) & 1-2 & 12-1 & 2.82e+08 \\ 
  CBVS-E (pessimistic) & 1-1 & 12-2 & 3.47e+08 \\ 
  Horseshoe & 1 & 15 & 2.95e+08 \\ 
  Spike and Slab & 6 & 10 & 2.90e+08  \\ 
  Bayesian LASSO & 2 & 14 & 2.90e+08  \\ 
  MCP & 4 & 12 & 2.86e+08 \\ 
  SCAD & 3 & 13 & 2.79e+08  \\ 
  LASSO & 6 & 10 & 2.77e+08  \\ 
  \hline
    \multicolumn{4}{|c|}{Lymphoma Dataset (p = 200, n = 100 }\\
  \hline
CBVS-NV (optimistic) & 3-0 & 0-197 & 2.54e+02\\ 
  CBVS-NV (pessimistic) & 3-36 & 0-161 & 8.98e+02 \\ 
  CBVS-E (optimistic) & 3-1 & 196-0 & 2.55e+02 \\ 
  CBVS-E (pessimistic) & 3-0 & 196-1 & 2.99e+02 \\ 
  Horseshoe & 0 & 200 &3.36e+02 \\ 
  Spike and Slab & 12 & 188 & 2.94e+02  \\ 
  Bayesian LASSO & 3 & 197 & 2.70e+02 \\ 
  MCP & 12 & 188 & \textbf{2.26e+02} \\ 
  SCAD & 21 & 179 & 2.45e+02 \\ 
  LASSO & 27 & 173 & {2.29e+02} \\ 
\hline
    \end{tabular}
    \caption{Summary of variable selection and model fitting for the real datasets.
    The first four rows in each section show the result of our robust
Bayesian analysis and therefore the numbers are separated by hyphens as some 
of the variables are indeterminate in our analyses but are active (inactive)
in the concerning model.}
\label{tab:real:sum}
\end{table}

\section{Conclusion}\label{sec:conc}
In this article, we propose a \addition{sensitivity analysis based cautious} 
variable selection \addition{scheme} for high dimensional problems 
\addition{using the} spike and slab framework.
Our framework is focused on the effect of prior elicitation in high dimensional
problems with limited information. We incorporate the prior information through
a set of priors and perform a robust Bayesian analysis \addition{by checking the sensitivity
of the variable selection over this set of priors.} The choice of conjugate
priors in our hierarchical model allows us to inspect several properties of the
posterior which are desirable in a robust Bayesian context. We provide the possibility
of controlling the prior on the regression coefficients through prior 
selection probabilities as well as the scale parameters. \addition{We discuss
the notion of a cautious variable selection rule based on this sensitivity analysis which
is robust to our choice of hyperparameter.} We also provide a suitable Gibbs sampling 
framework for the \addition{parameter estimation}. 

\addition{A major result which we obtain from our hierarchical model
is that the posterior odds of model selection are monotone with respect
to the prior expectations of the selection probabilities. This result shows
the importance of robust Bayesian analysis, especially for high dimensional 
problems with limited information, where extracting information on the
model size is extremely difficult. However, unlike the classical 
approach, robust Bayesian analysis will not give us the most probable 
model, instead it will return a set of most probable models.
We can further analyse this set based on the original sector of the 
problem. For instance, this is particularly useful 
in decision making problems, when there's a non-negative
gain for abstaining from making a decision \cite{du2005} and 
a wrongly selected model will lead to undesirable loss.
}

One of the important aspects of Bayesian variable selection methods is to 
investigate maximum a posteriori estimates which we deliberately ignore in this
work. We observe that, for obtaining sparse MAP estimates, we require suitable
regularity conditions on prior parameters to attain sparsity in an asymptotic
sense. These regularity conditions on the prior parameter go against our approach 
where we are more interested in prior elicitation to tackle the severe uncertainty
around the problem rather than using  automatic parameter tuning.

We analysed both synthetic datasets and real datasets to illustrate our
method. We considered different aspects of variable selection problems and 
constructed our synthetic datasets accordingly. We observe that our method
provides us more reliable results for dense models which is often not the case
for other methods used for comparison. We also realised the need of a more
convincing utility based loss function for model comparison. Currently, we are relying
on \addition{the squared error to obtain an optimistic fit and a pessimistic fit.
This is somewhat related to posterior cross-validation, especially if we only consider
the optimistic fit. However, choosing the single best model may not be beneficial
in every situation, which we showed in the illustration with the first synthetic dataset.
As we propose this sensitivity analysis over the hyper parameter to obtain a cautious
variable selection routine, we also get some indeterminacy between the multiple model
which we discuss with a measure called indeterminacy.} These measures give us an
overview of the \addition{uncertainty as well as the goodness in model fitting. But, 
comparison with other methods can be difficult at times as we do not have a 
unified measure to capture goodness and model indeterminacy simultaneously. 
However, some may argue that in specific cases only the optimistic fit 
is enough for comparison.}

\appendix

\addition{\section{Orthogonal Design Case}\label{app:sec:odc}}

\begin{lemma}
Let,
\begin{equation}
    \beta_j \mid y
    \sim\frac{\alpha_jw_{1,j}}{W_j}
    \mathcal{N}\left(\hat{\beta}_{1,j}, \sigma_1^2\right)
    +\frac{(1-\alpha_j)w_{0,j}}{W_j}
    \mathcal{N}\left(\hat{\beta}_{0,j}, \sigma_0^2\right).
\end{equation}
Then, $\mathcal{N}\left(\hat{\beta}_{1,j}, \sigma_1^2\right)$ dominates the posterior if
\begin{align}
    \hat{\beta_j^2} & >
    \frac{\sigma^2}{n}\frac{(n\tau_1^2+1)(n\tau_0^2+1)}{n\tau_1^2-n\tau_0^2}
    \left[
    2\ln\left(\frac{1-\epsilon_1}{\epsilon_1}\right)
    +\ln\left(\frac{n\tau_1^2+1}{n\tau_0^2+1}\right)\right]
\end{align}
and $\mathcal{N}\left(\hat{\beta}_{0,j}, \sigma_0^2\right)$ 
dominates the posterior if, 
\begin{align}
    \hat{\beta_j^2} & <
    \frac{\sigma^2}{n}\frac{(n\tau_1^2+1)(n\tau_0^2+1)}{n\tau_1^2-n\tau_0^2}
    \left[
    2\ln\left(\frac{\epsilon_2}{1-\epsilon_2}\right)
    +\ln\left(\frac{n\tau_1^2+1}{n\tau_0^2+1}\right)\right].
\end{align}

\end{lemma}
\begin{proof}\label{ap:lem:2}
Exploiting the monotonicity property of the posterior odds, we can
say that $\mathcal{N}\left(\hat{\beta}_{1,j}, \sigma_1^2\right)$ dominates the 
posterior if $\frac{\epsilon_1}{(1-\epsilon_1)}\cdot \frac{w_{1,j}}{w_{0,j}}>1$.
That is, if
\begin{align}
    \exp\left(-\frac{n\hat{\beta_j^2}}{2(n\sigma^2\tau_1^2+\sigma^2)}
    +\frac{n\hat{\beta_j^2}}{2(n\sigma^2\tau_0^2+\sigma^2)}\right)
    &>\frac{(1-\epsilon_1)\sqrt{n\tau_1^2+1}}{\epsilon_1\sqrt{n\tau_0^2+1}}\\
    -\frac{n\hat{\beta_j^2}}{2(n\sigma^2\tau_1^2+\sigma^2)}
    +\frac{n\hat{\beta_j^2}}{2(n\sigma^2\tau_0^2+\sigma^2)}
    &>\ln\left(
    \frac{(1-\epsilon_1)\sqrt{n\tau_1^2+1}}{\epsilon_1\sqrt{n\tau_0^2+1}}\right)\\
    \frac{n\hat{\beta_j^2}}{2\sigma^2}
    \left[-\frac{1}{(n\tau_1^2+1)}+\frac{1}{(n\tau_0^2+1)}\right]
    &>\ln\left(
    \frac{(1-\epsilon_1)\sqrt{n\tau_1^2+1}}{\epsilon_1\sqrt{n\tau_0^2+1}}\right)\\
    \frac{n\hat{\beta_j^2}}{2\sigma^2}
    \frac{n\tau_1^2-n\tau_0^2}{(n\tau_1^2+1)(n\tau_0^2+1)}
    &>\ln\left(
    \frac{(1-\epsilon_1)\sqrt{n\tau_1^2+1}}{\epsilon_1\sqrt{n\tau_0^2+1}}\right).
\end{align}
Then after rearranging the terms on both sides, we get:
\begin{align}
    \hat{\beta_j^2} & >
    \frac{\sigma^2}{n}\frac{(n\tau_1^2+1)(n\tau_0^2+1)}{n\tau_1^2-n\tau_0^2}
    \left[
    2\ln\left(\frac{1-\epsilon_1}{\epsilon_1}\right)
    +\ln\left(\frac{n\tau_1^2+1}{n\tau_0^2+1}\right)\right].
\end{align}
Similarly we can prove the other inequality.
\end{proof}

\begin{lemma}\label{ap:lem:3}
Let $g:(0,1)\to \mathbb{R}$ be defined as 
\begin{equation}
    g(\alpha) \coloneqq \frac{a\alpha + b}{c\alpha + d}
\end{equation}
for some constants $a$, $b$, $c$, and $d\in\mathbb{R}$
such that $c\alpha+d>0\quad \forall \alpha \in (0,1)$.
Then $g$ is monotonically increasing when $ad-bc > 0$ and 
monotonically decreasing when $ad-bc<0$.
\end{lemma}

\begin{proof}
The proof is straightforward and therefore we omit it.
\end{proof}
Note that, the posterior mean of $\beta_j$ can be written as:
\begin{align}
    E(\beta_j\mid y) 
    & = \frac{(w_{1,j}\hat{\beta}_{1,j} - w_{0,j}\hat{\beta}_{1,j})\alpha_j + w_{0,j}\hat{\beta}_{0,j}}
    {(w_{1,j}-w_{0,j})\alpha_j + w_{0,j}}.
\end{align}
Therefore, we have
\begin{align}
    \frac{d}{d\alpha_j}E(\beta_j\mid y) &= \frac{w_{0, j}w_{1, j}
    (\hat{\beta}_{1,j}-\hat{\beta}_{0,j})}{[(w_{1,j}-w_{0,j})\alpha_j + w_{0,j}]^2}\\
    &= \frac{w_{0, j}w_{1, j}}
    {[(w_{1,j}-w_{0,j})\alpha_j + w_{0,j}]^2}
    \left(\frac{n\tau_1^2\hat{\beta_j}}{n\tau_1^2 +1}
    -\frac{n\tau_0^2\hat{\beta_j}}{n\tau_0^2 +1}\right)\\
    &= \frac{w_{0, j}w_{1, j}\hat{\beta}_j}
    {[(w_{1,j}-w_{0,j})\alpha_j + w_{0,j}]^2}
    \left(\frac{n\tau_1^2}{n\tau_1^2 +1} - \frac{n\tau_0^2}{n\tau_0^2 +1}\right)\\
    &= \frac{w_{0, j}w_{1, j}}
    {[(w_{1,j}-w_{0,j})\alpha_j + w_{0,j}]^2}
    \frac{n\tau_1^2 - n\tau_0^2}{(n\tau_1^2 +1)(n\tau_0^2 +1)}\hat{\beta}_j.
\end{align}
Since $\tau_1>\tau_0$, the posterior mean 
of $\beta_j$ is monotonically increasing with respect to 
$\alpha_j$ for $\hat{\beta}_j > 0$ and monotonically decreasing with respect to
$\alpha_j$ for $\hat{\beta}_j < 0$.

\begin{lemma}\label{ap:lem:4}
Let 
\begin{equation}
    X\sim w_1 f_1 + w_2 f_2
\end{equation}
where $f_i$ denotes a normal density with mean
$\mu_i$ and variance $\sigma^2_i$ for $i=1,2$. Then,
\begin{equation}
    Var(X) = \sum_{i=1}^2 w_i(\sigma^2_i + \mu_i^2) - 
    \left(\sum_{i=1}^2 w_i\mu_i\right)^2
\end{equation}
\end{lemma}
\begin{proof}
First, note that
\begin{align}
    E(X^2) &=\int x^2[w_1f_1(x) + w_2 f_2(x)]dx\\
    & = w_1\int x^2f_1(x)dx + w_2\int x^2f_2(x)dx\\
    & = w_1(\sigma^2_1 + \mu_1^2) + w_2(\sigma^2_2 + \mu_2^2).
\end{align}
Consequently, Then, the variance of $X$ is given by:
\begin{align}
    \text{Var}(X) &= E(X^2) - \left[E(X)\right]^2\\
    & = \sum_{i=1}^2w_i(\sigma^2_i + \mu_i^2) - \left(\sum_{i=1}^2 w_i\mu_i\right)^2.
\end{align}
\end{proof}
Now, we know that, 
\begin{align}
    \beta_j \mid y
    &\sim\frac{\alpha_jw_{1,j}}{W_j}
    \mathcal{N}\left(\hat{\beta}_{1,j}, \sigma_1^2\right)
    +\frac{(1-\alpha_j)w_{0,j}}{W_j}
    \mathcal{N}\left(\hat{\beta}_{0,j}, \sigma_0^2\right).
\end{align}
Then from above lemma, we can show that the variance of $\beta_j\mid y$ is given by:
\begin{align}
    \text{Var}(\beta_j\mid y)
    &= \frac{\alpha_jw_{1,j}\sigma^2_1 + (1-\alpha_j)w_{0,j}\sigma^2_0}{W_j}
    + \frac{\alpha(1-\alpha)w_{1,j}w_{0,j}(\hat{\beta}_{1,j} - \hat{\beta}_{0,j})^2}
    {W_j^2}.
\end{align}

\addition{

\section{Regression Coefficients}\label{app:sec:rc}

\begin{align}
    P(\beta\mid y)
    &\overset{\beta}{\propto}\int \int \sum_{\gamma} P(\beta, \gamma, \sigma^2, q \mid y) dq d\sigma^2\\
    &\overset{\beta}{\propto}\int P(y\mid \beta, \sigma^2)
    \sum_{\gamma}P(\beta \mid \gamma, \sigma^2)
    \left[\int P(\gamma\mid q) P(q) dq\right]
    P(\sigma^2) d \sigma^2\\
    &\overset{\beta}{\propto}\int P(y\mid \beta, \sigma^2)  
    \prod_{j}\left[\alpha_jf_1(\beta_j)+ (1-\alpha_j)f_0(\beta_j)\right]
    P(\sigma^2) d \sigma^2\label{eq:beta:joint:post}
\end{align}
To simplify the product term in the above expression, we first provide the following
identity:
\begin{lemma}\label{lem:prod:expand}
Let,
\begin{equation}
    f_{\gamma_j}(\beta_j) \coloneqq \frac{1}{\sqrt{2\pi\sigma^2\tau_1^{2\gamma_j}\tau_0^{2(1-\gamma_j)}}}
    \exp\left(-\frac{\beta_j^2}{2\sigma^2\tau_1^{2\gamma_j}\tau_0^{2(1-\gamma_j)}}\right).
\end{equation}
Then,
\begin{equation}
    \prod_{j}\left[\alpha_jf_1(\beta_j)+ (1-\alpha_j)f_0(\beta_j)\right] = 
    \frac{1}{\sqrt{(2\pi\sigma^2\tau_1^2)^p}}\exp\left(-\frac{\|\beta\|^2}{2\sigma^2\tau_1^2}\right)
    \left(\prod_{j}\alpha_j\right)\left(\sum_{k=0}^pg_{k}(\beta, \sigma^2)\right) ,
\end{equation}
where 
\begin{align}
    g_0(\beta, \sigma^2) &= 1,\\
    g_p(\beta, \sigma^2) &= \tfrac{\tau_1^{p}}{\tau_0^{p}}\exp\left(-\frac{\|\beta\|^2}{2\sigma^2}
    \left(\frac{1}{\tau_0^2} -\frac{1}{\tau_1^2}\right)\right)
    \prod_{j}\frac{1-\alpha_{j}}{\alpha_{j}},
    \intertext{and for $1\le k \le p-1$,}
    g_k(\beta, \sigma^2) &= \tfrac{\tau_1^{k}}{\tau_0^{k}}
    \sum_{j_1<\cdots<j_{k}}\tfrac{1-\alpha_{j_1}}{\alpha_{j_1}}\tfrac{1-\alpha_{j_2}}{\alpha_{j_2}}
    \cdots\tfrac{1-\alpha_{j_{k}}}{\alpha_{j_{k}}}
    \exp\left(-\tfrac{\beta_{j_1}^2+\beta_{j_2}^2+\cdots+\beta_{j_{k}}^2}{2\sigma^2}
    \left(\tfrac{1}{\tau_0^2} -\tfrac{1}{\tau_1^2}\right)\right).
\end{align}
\end{lemma}

\begin{proof} \label{ap:lem:5}
From the left hand side, we have
\begin{align}
    &\prod_{j}\left[\alpha_jf_1(\beta_j)+ (1-\alpha_j)f_0(\beta_j)\right]\nonumber\\
    &=\prod_{j}\alpha_jf_1(\beta_j)\prod_{j}\left[1+ \frac{1-\alpha_j}{\alpha_j}
    \frac{f_0(\beta_j)}{f_1(\beta_j)}\right]\\
    &=\prod_{j}\frac{\alpha_j}{\sqrt{2\pi\sigma^2\tau_1^2}}
    \exp\left(-\frac{\beta_j^2}{2\sigma^2\tau_1^2}\right)\prod_{j}
    \left[1+ \frac{\tau_1(1-\alpha_j)}{\tau_0\alpha_j}
    \exp\left(-\frac{\beta_j^2}{2\sigma^2}\left(\frac{1}{\tau_0^2} -\frac{1}{\tau_1^2}\right)\right)\right]\\
    &=\frac{1}{\sqrt{(2\pi\sigma^2\tau_1^2)^p}}\exp\left(-\frac{\|\beta\|^2}{2\sigma^2\tau_1^2}\right)
    \prod_{j}\alpha_j\prod_{j}
    \left[1+ \frac{\tau_1(1-\alpha_j)}{\tau_0\alpha_j}
    \exp\left(-\frac{\beta_j^2}{2\sigma^2}\left(\frac{1}{\tau_0^2} -\frac{1}{\tau_1^2}\right)\right)\right]\label{eq:prod:mixture}
\end{align}
Now to evaluate the product term $\prod_{j}
    \left[1+ \frac{\tau_1(1-\alpha_j)}{\tau_0\alpha_j}
    \exp\left(-\frac{\beta_j^2}{2\sigma^2}\left(\frac{1}{\tau_0^2} -\frac{1}{\tau_1^2}\right)\right)\right]$, we use the following identity
\begin{align}
    &\prod_{j}(1+a_j) \nonumber\\
    &= 1 + \sum_{j}a_j + \sum_{j_1<j_2}a_{j_1}a_{j_2} + 
    \sum_{j_1<j_2<j_3}a_{j_1}a_{j_2}a_{j_3} + \cdots + \sum_{j_1<\cdots<j_{(p-1)}}(a_{j_1}
    \cdots a_{j_{(p-1)}})+ \prod_{j}a_j.\label{eq:identity:prod}
\end{align}
Then, combining \cref{eq:prod:mixture} and \cref{eq:identity:prod}, we have
\begin{align}
    &\prod_{j}    \left[1+ \frac{\tau_1(1-\alpha_j)}{\tau_0\alpha_j}
    \exp\left(-\frac{\beta_j^2}{2\sigma^2}\left(\frac{1}{\tau_0^2} 
    -\frac{1}{\tau_1^2}\right)\right)\right]\nonumber\\
    &=1 + \frac{\tau_1}{\tau_0}\sum_{j}\frac{1-\alpha_j}{\alpha_j}
    \exp\left(-\frac{\beta_j^2}{2\sigma^2}\left(\frac{1}{\tau_0^2} 
    -\frac{1}{\tau_1^2}\right)\right) \nonumber\\
    & + \frac{\tau_1^2}{\tau_0^2}
    \sum_{j_1<j_2}\frac{1-\alpha_{j_1}}{\alpha_{j_1}}\frac{1-\alpha_{j_2}}{\alpha_{j_2}}
    \exp\left(-\frac{\beta_{j_1}^2+\beta_{j_2}^2}{2\sigma^2}\left(\frac{1}{\tau_0^2} 
    -\frac{1}{\tau_1^2}\right)\right) \nonumber\\
    & + \frac{\tau_1^3}{\tau_0^3}
    \sum_{j_1<j_2<j_3}\frac{1-\alpha_{j_1}}{\alpha_{j_1}}\frac{1-\alpha_{j_2}}{\alpha_{j_2}}
    \frac{1-\alpha_{j_3}}{\alpha_{j_3}}
    \exp\left(-\frac{\beta_{j_1}^2+\beta_{j_2}^2+\beta_{j_3}^2}{2\sigma^2}\left(\frac{1}{\tau_0^2} 
    -\frac{1}{\tau_1^2}\right)\right) \nonumber\\
    & + \cdots\nonumber\\
    & + \frac{\tau_1^{(p-1)}}{\tau_0^{(p-1)}}
    \sum_{j_1<\cdots<j_{p-1}}\tfrac{1-\alpha_{j_1}}{\alpha_{j_1}}\tfrac{1-\alpha_{j_2}}{\alpha_{j_2}}
    \cdots\tfrac{1-\alpha_{j_{(p-1)}}}{\alpha_{j_{(p-1)}}}
    \exp\left(-\frac{\beta_{j_1}^2+\beta_{j_2}^2+\cdots+\beta_{j_{(p-1)}}^2}{2\sigma^2}\left(\frac{1}{\tau_0^2} 
    -\frac{1}{\tau_1^2}\right)\right)\nonumber\\
    & + \frac{\tau_1^{p}}{\tau_0^{p}}\exp\left(-\frac{\|\beta\|^2}{2\sigma^2}\left(\frac{1}{\tau_0^2} 
    -\frac{1}{\tau_1^2}\right)\right)
    \prod_{j}\frac{1-\alpha_{j}}{\alpha_{j}}.
\end{align}
Therefore,
\begin{equation}
    \prod_{j}\left[\alpha_jf_1(\beta_j)+ (1-\alpha_j)f_0(\beta_j)\right] = 
    \frac{1}{\sqrt{(2\pi\sigma^2\tau_1^2)^p}}\exp\left(-\frac{\|\beta\|^2}{2\sigma^2\tau_1^2}\right)
    \left(\prod_{j}\alpha_j\right)\left(\sum_{k=0}^pg_{k}(\beta, \sigma^2)\right) 
\end{equation}
where $g_k(\beta, \sigma^2)$ is as specified before.
\end{proof}

Now, using our result from \cref{lem:prod:expand} in \cref{eq:beta:joint:post}, we get
\begin{align}
    &P(\beta\mid y)\nonumber\\
    &\overset{\beta}{\propto}\int P(y\mid \beta, \sigma^2) 
    \frac{1}{\sqrt{(2\pi\sigma^2\tau_1^2)^p}}
    \exp\left(-\frac{\|\beta\|^2}{2\sigma^2\tau_1^2}\right)
    \prod_{j}\alpha_j\left(\sum_{k=0}^pg_{k}(\beta, \sigma^2)\right)
    P(\sigma^2)d\sigma^2\\
    &\overset{\beta}{\propto}\sum_{k=0}^p\int\frac{1}{\sigma^{2(n/2+p/2)}}
    \exp\left(-\frac{1}{\sigma^2}
    \left(\frac{\|y-\mathbf{x}\beta\|^2_2}{2}
    +\frac{\|\beta\|^2}{2\tau_1^2}\right)\right)
    g_k(\beta, \sigma^2)\frac{1}{\sigma^{2(a+1)}}
    \exp\left(-\frac{b}{\sigma^2}\right)d\sigma^2\\
    &\overset{\beta}{\propto} \sum_{k=0}^p\int\frac{1}{\sigma^{2(n/2+p/2 + a +1)}}
    \exp\left(-\frac{1}{\sigma^2}\left(\frac{\|y-\mathbf{x}\beta\|^2_2}{2}
    +\frac{\|\beta\|^2}{2\tau_1^2} + b\right)\right)g_k(\beta, \sigma^2)d\sigma^2
    \label{eq:beta:joint:post:before:t}
\end{align}
Before evaluating the integrals, we need to show some identities. For that, we first
need to define some expressions. Let, $D_{\tau_1} = \tau_1^{-2}\mathbf{I}_{p}$,
then we define
\begin{align}
    L_{\tau_1} = (\mathbf{x}^T\mathbf{x} + D_{\tau_1})^{-1}, 
    \mu_{\tau_1} = L_{\tau_1}\mathbf{x}^Ty, 
    r_{\tau_1} = \frac{y^Ty - y^T\mathbf{x}L_{\tau_1}\mathbf{x}^Ty}{2} +b, \text{ and }
    \Sigma^{-1}_{\tau_1} = \frac{n+2a}{2r_{\tau_1}}L_{\tau_1}^{-1}.
\end{align}
We also use similar expressions using $D_{\tau_0}$ and $D_{j_1,j_2,\cdots,j_k}$
where
\begin{equation}
    D_{\tau_0} = \tau_0^{-2}\mathbf{I}_{p} \text{ and }
    D_{j_1,j_2,\cdots,j_k} = \text{diag}((1-\mathbb{I}_{j_1,j_2,\cdots,j_k}(j))\tau_1^{-2} 
    +\mathbb{I}_{j_1,j_2,\cdots,j_k}(j)\tau_0^{-2}).
\end{equation}

\begin{lemma}\label{lem:int:t:dist}
Let $a^*/2 = n/2+a$. Then,
\begin{equation}
    \int\frac{1}{\sigma^{2(n/2+p/2 + a +1)}}
    \exp\left(-\frac{1}{\sigma^2}\left(\frac{\|y-\mathbf{x}\beta\|^2_2}{2}
    +\frac{\|\beta\|^2}{2\tau_1^2} + b\right)\right)g_k(\beta, \sigma^2)d\sigma^2\\
    = \Gamma(a^*/2)(a^*\pi)^{p/2}h_k(\beta)
\end{equation}
where,
\begin{align}
    h_0(\beta) &= \frac{\sqrt{|\Sigma_{\tau_1}|}}
    {r_{\tau_1}^{a^*/2 + p/2}}
    \mathcal{T}_{a^*}(\mu_{\tau_1}, \Sigma_{\tau_1})\\
    h_p(\beta) &= \frac{\tau_1^{p}}{\tau_0^{p}}\prod_{j}\frac{1-\alpha_{j}}{\alpha_{j}}
    \frac{\sqrt{|\Sigma_{\tau_0}|}}
    {r_{\tau_0}^{a^*/2 + p/2}}
    \mathcal{T}_{a^*}(\mu_{\tau_0}, \Sigma_{\tau_0})
    \intertext{and for $1\le k \le p-1$}
    h_k(\beta) & = \frac{\tau_1^{k}}{\tau_0^{k}}
    \sum_{j_1<\cdots<j_{k}}\frac{1-\alpha_{j_1}}{\alpha_{j_1}}\frac{1-\alpha_{j_2}}{\alpha_{j_2}}
    \cdots\frac{1-\alpha_{j_{k}}}{\alpha_{j_{k}}}
    \frac{\sqrt{|\Sigma_{j_1,j_2,\cdots,j_k}|}}
    {r_{j_1,j_2,\cdots,j_k}^{a^*/2 + p/2}}
    \mathcal{T}_{a^*}(\mu_{j_1,j_2,\cdots,j_k}, \Sigma_{j_1,j_2,\cdots,j_k})
\end{align}
\end{lemma}

\begin{proof} \label{ap:lem:6}
We compute the integrals using the properties of inverse gamma
distribution followed by some adjustments to obtain the expression
of multivariate t distribution.

For $k=0$, we have
\begin{align}
    &\int\frac{1}{\sigma^{2(n/2+p/2 + a +1)}}
    \exp\left(-\frac{1}{\sigma^2}\left(\frac{\|y-\mathbf{x}\beta\|^2_2}{2}
    +\frac{\|\beta\|^2}{2\tau_1^2} + b\right)\right)g_0(\beta, \sigma^2)d\sigma^2\nonumber\\
    & = \int\frac{1}{\sigma^{2(a^*/2 + p/2 +1)}}
    \exp\left(-\frac{1}{\sigma^2}\left(\frac{\|y-\mathbf{x}\beta\|^2_2}{2}
    +\frac{\|\beta\|^2}{2\tau_1^2} + b\right)\right)d\sigma^2\\
    &= \frac{\Gamma(a^*/2+p/2)}{\left(\frac{\|y-\mathbf{x}\beta\|^2_2}{2}
    +\frac{\|\beta\|^2}{2\tau_1^2} +b\right)^{a^*/2 + p/2}}\\
    &= \frac{\Gamma(a^*/2+p/2)}{\left(\frac{\|y-\mathbf{x}\beta\|^2_2}{2}
    +\frac{\beta^t D_{\tau_1}\beta}{2} +b\right)^{a^*/2 + p/2}}\\
    &= \frac{\Gamma(a^*/2+p/2)}{\left(\frac{(\beta-\mu_{\tau_1})^TL_{\tau_1}^{-1}
    (\beta-\mu_{\tau_1})}{2}
    +r_{\tau_1}\right)^{a^*/2 + p/2}}\\
    &= \frac{\Gamma(a^*/2+p/2)}{r_{\tau_1}^{a^*/2 + p/2}
    \left(\frac{1}{a^*}\frac{(\beta-\mu_{\tau_1})^TL_{\tau_1}^{-1}
    (\beta-\mu_{\tau_1})}{2r_{\tau_1}/a^*}
    +1\right)^{a^*/2 + p/2}}\\
    &= \frac{\Gamma(a^*/2+p/2)}{r_{\tau_1}^{a^*/2 + p/2}
    \left(1+\frac{1}{a^*}(\beta-\mu_{\tau_1})^T 
    \Sigma^{-1}_{\tau_1}(\beta-\mu_{\tau_1})\right)^{a^*/2 + p/2}}\\
    &= \tfrac{\Gamma(a^*/2)(a^*\pi)^{p/2}\sqrt{|\Sigma_{\tau_1}|}}
    {r_{\tau_1}^{a^*/2 + p/2}}
    \tfrac{\Gamma(a^*/2+p/2)}{\Gamma(a^*/2)(a^*\pi)^{p/2}\sqrt{|\Sigma_{\tau_1}|}
    \left(1+\frac{1}{a^*}(\beta-\mu_{\tau_1})^T 
    \Sigma^{-1}_{\tau_1}(\beta-\mu_{\tau_1})\right)^{a^*/2 + p/2}}\\
    &= \Gamma(a^*/2)(a^*\pi)^{p/2}\frac{\sqrt{|\Sigma_{\tau_1}|}}
    {r_{\tau_1}^{a^*/2 + p/2}}
    \mathcal{T}_{a^*}(\mu_{\tau_1}, \Sigma_{\tau_1})\\
    & = \Gamma(a^*/2)(a^*\pi)^{p/2} h_0(\beta).
\end{align}

For $k=p$, we have
\begin{align}
    &\int\frac{1}{\sigma^{2(n/2+p/2 + a +1)}}
    \exp\left(-\frac{1}{\sigma^2}\left(\frac{\|y-\mathbf{x}\beta\|^2_2}{2}
    +\frac{\|\beta\|^2}{2\tau_1^2} + b\right)\right)g_p (\beta,\sigma^2)d\sigma^2\nonumber\\
    & = \int\frac{1}{\sigma^{2(a^*/2 + p/2 +1)}}
    \exp\left(-\frac{1}{\sigma^2}\left(\frac{\|y-\mathbf{x}\beta\|^2_2}{2}
    +\frac{\|\beta\|^2}{2\tau_1^2} + b\right)\right)
    \\*\nonumber&\qquad\qquad\frac{\tau_1^{p}}{\tau_0^{p}}\exp\left(-\frac{\|\beta\|^2}{2\sigma^2}\left(\frac{1}{\tau_0^2} 
    -\frac{1}{\tau_1^2}\right)\right)
    \prod_{j}\frac{1-\alpha_{j}}{\alpha_{j}}d\sigma^2\\
    & = \frac{\tau_1^{p}}{\tau_0^{p}}\prod_{j}\frac{1-\alpha_{j}}{\alpha_{j}}
    \int\frac{1}{\sigma^{2(a^*/2 + p/2 +1)}}
    \exp\left(-\frac{1}{\sigma^2}\left(\frac{\|y-\mathbf{x}\beta\|^2_2}{2}
    +\frac{\|\beta\|^2}{2\tau_0^2} +b \right)\right)\\
    & = \frac{\tau_1^{p}}{\tau_0^{p}}\prod_{j}\frac{1-\alpha_{j}}{\alpha_{j}}
    \frac{\Gamma(a^*/2+p/2)}{\left(\frac{\|y-\mathbf{x}\beta\|^2_2}{2}
    +\frac{\|\beta\|^2}{2\tau_0^2}+b\right)^{a^*/2 + p/2}}\\
    & = \frac{\tau_1^{p}}{\tau_0^{p}}\prod_{j}\frac{1-\alpha_{j}}{\alpha_{j}}
    \frac{\Gamma(a^*/2+p/2)}{\left(\frac{\|y-\mathbf{x}\beta\|^2_2}{2}
    +\frac{\beta^T D_{\tau_0}\beta}{2}+b\right)^{a^*/2 + p/2}}\\
    & = \frac{\tau_1^{p}}{\tau_0^{p}}\prod_{j}\frac{1-\alpha_{j}}{\alpha_{j}}
    \frac{\Gamma(a^*/2+p/2)}{r_{\tau_0}^{a^*/2 + p/2}
    \left(1+\frac{1}{a^*}(\beta-\mu_{\tau_0})^T 
    \Sigma^{-1}_{\tau_0}(\beta-\mu_{\tau_0})\right)^{a^*/2 + p/2}}\\
    & = \frac{\tau_1^{p}}{\tau_0^{p}}\prod_{j}\frac{1-\alpha_{j}}{\alpha_{j}}
    \tfrac{\Gamma(a^*/2)(a^*\pi)^{p/2}\sqrt{|\Sigma_{\tau_0}|}}
    {r_{\tau_0}^{a^*/2 + p/2}}
    \tfrac{\Gamma(a^*/2+p/2)}{\Gamma(a^*/2)(a^*\pi)^{p/2}\sqrt{|\Sigma_{\tau_0}|}
    \left(1+\frac{1}{a^*}(\beta-\mu_{\tau_0})^T 
    \Sigma^{-1}_{\tau_0}(\beta-\mu_{\tau_0})\right)^{a^*/2 + p/2}}\\
    &= \Gamma(a^*/2)(a^*\pi)^{p/2}\frac{\tau_1^{p}}{\tau_0^{p}}\prod_{j}\frac{1-\alpha_{j}}{\alpha_{j}}
    \frac{\sqrt{|\Sigma_{\tau_0}|}}
    {r_{\tau_0}^{a^*/2 + p/2}}
    \mathcal{T}_{a^*}(\mu_{\tau_0}, \Sigma_{\tau_0})\\
    & = \Gamma(a^*/2)(a^*\pi)^{p/2} h_p(\beta).
\end{align}

For $1\le k \le p-1$, we have 
\begin{align}
    & \int\frac{1}{\sigma^{2(n/2+p/2 + a +1)}}
    \exp\left(-\frac{1}{\sigma^2}\left(\frac{\|y-\mathbf{x}\beta\|^2_2}{2}
    +\frac{\|\beta\|^2}{2\tau_1^2} + b\right)\right)g_k(\beta, \sigma^2)d\sigma^2\nonumber\\
    & = \int\frac{1}{\sigma^{2(a^*/2+p/2 +1)}}
    \exp\left(-\frac{1}{\sigma^2}\left(\frac{\|y-\mathbf{x}\beta\|^2_2}{2}
    +\frac{\|\beta\|^2}{2\tau_1^2} + b\right)\right)\nonumber\\*
    &\qquad\qquad\tfrac{\tau_1^{k}}{\tau_0^{k}}
    \sum_{j_1<\cdots<j_{k}}\tfrac{1-\alpha_{j_1}}{\alpha_{j_1}}\tfrac{1-\alpha_{j_2}}{\alpha_{j_2}}
    \cdots\tfrac{1-\alpha_{j_{k}}}{\alpha_{j_{k}}}
    \exp\left(-\tfrac{\beta_{j_1}^2+\beta_{j_2}^2+\cdots+\beta_{j_{k}}^2}{2\sigma^2}\left(\tfrac{1}{\tau_0^2} 
    -\tfrac{1}{\tau_1^2}\right)\right)d\sigma^2\\
    &= \frac{\tau_1^{k}}{\tau_0^{k}}
    \sum_{j_1<\cdots<j_{k}}\frac{1-\alpha_{j_1}}{\alpha_{j_1}}\frac{1-\alpha_{j_2}}{\alpha_{j_2}}
    \cdots\frac{1-\alpha_{j_{k}}}{\alpha_{j_{k}}}\nonumber\\*
    & \qquad\qquad\int\frac{1}{\sigma^{2(a^*/2 + p/2 +1)}}
    \exp\left(-\frac{1}{\sigma^2}\left(\frac{\|y-\mathbf{x}\beta\|^2_2}{2}
    +\frac{\|\beta\|^2}{2\tau_1^2}+b\right)\right)
    \nonumber\\*&\qquad\qquad\qquad\qquad\exp\left(-\frac{\beta_{j_1}^2+\beta_{j_2}^2+\cdots+\beta_{j_{k}}^2}{2\sigma^2}\left(\frac{1}{\tau_0^2} 
    -\frac{1}{\tau_1^2}\right)\right)d\sigma^2\\
    &= \frac{\tau_1^{k}}{\tau_0^{k}}
    \sum_{j_1<\cdots<j_{k}}\frac{1-\alpha_{j_1}}{\alpha_{j_1}}\frac{1-\alpha_{j_2}}{\alpha_{j_2}}
    \cdots\frac{1-\alpha_{j_{k}}}{\alpha_{j_{k}}}\nonumber\\*
    &\qquad\qquad \int\frac{1}{\sigma^{2(a^*/2 + p/2 +1)}}
    \exp\left(-\frac{1}{\sigma^2}\left(\frac{\|y-\mathbf{x}\beta\|^2_2}{2}
    +\frac{\beta^T D_{j_1,j_2,\cdots,j_k}\beta}{2} + b\right)\right)d\sigma^2\\
    &= \frac{\tau_1^{k}}{\tau_0^{k}}
    \sum_{j_1<\cdots<j_{k}}\frac{1-\alpha_{j_1}}{\alpha_{j_1}}\frac{1-\alpha_{j_2}}{\alpha_{j_2}}
    \cdots\frac{1-\alpha_{j_{k}}}{\alpha_{j_{k}}}
    \frac{\Gamma(a^*/2 + p/2)}{
    \left(\frac{\|y-\mathbf{x}\beta\|^2_2}{2}
    +\frac{\beta^T D_{j_1,j_2,\cdots,j_k}\beta}{2} + b\right)^{a^*/2+p/2}}\\
    &= \frac{\tau_1^{k}}{\tau_0^{k}}
    \sum_{j_1<\cdots<j_{k}}\frac{1-\alpha_{j_1}}{\alpha_{j_1}}\frac{1-\alpha_{j_2}}{\alpha_{j_2}}
    \cdots\frac{1-\alpha_{j_{k}}}{\alpha_{j_{k}}}
    \frac{\Gamma(a^*/2)(a^*\pi)^{p/2}\sqrt{|\Sigma_{j_1,j_2,\cdots,j_k}|}}
    {r_{j_1,j_2,\cdots,j_k}^{a^*/2 + p/2}}\nonumber\\*
    &
    \qquad\qquad\tfrac{\Gamma(a^*/2+p/2)}{\Gamma(a^*/2)(a^*\pi)^{p/2}
    \sqrt{|\Sigma_{j_1,j_2,\cdots,j_k}|}
    \left(1+\frac{1}{a^*}(\beta-\mu_{j_1,j_2,\cdots,j_k})^T 
    \Sigma^{-1}_{j_1,j_2,\cdots,j_k}(\beta-\mu_{j_1,j_2,\cdots,j_k})\right)^{a^*/2 + p/2}}\\
    &= \Gamma(a^*/2)(a^*\pi)^{p/2}\frac{\tau_1^{k}}{\tau_0^{k}}
    \nonumber\\*&\qquad\qquad\sum_{j_1<\cdots<j_{k}}\tfrac{1-\alpha_{j_1}}{\alpha_{j_1}}\tfrac{1-\alpha_{j_2}}{\alpha_{j_2}}
    \cdots\tfrac{1-\alpha_{j_{k}}}{\alpha_{j_{k}}}
    \tfrac{\sqrt{|\Sigma_{j_1,j_2,\cdots,j_k}|}}
    {r_{j_1,j_2,\cdots,j_k}^{a^*/2 + p/2}}
    \mathcal{T}_{a^*}(\mu_{j_1,j_2,\cdots,j_k}, \Sigma_{j_1,j_2,\cdots,j_k})\\
    &=\Gamma(a^*/2)(a^*\pi)^{p/2} h_k(\beta).
\end{align}
\end{proof}

Then, using identities from \cref{lem:int:t:dist},
we have
\begin{align}
    P(\beta \mid y) &\overset{\beta}{\propto} 
    \Gamma(a^*/2)(a^*\pi)^{p/2}\sum_{k=0}^p h_k(\beta)\\
    &\overset{\beta}{\propto} \sum_{k=0}^p h_k(\beta)
\end{align}
Now, for $1\le k \le p-1$, we can rewrite $h(k)$ so that
\begin{equation}
    h(k) = \sum_{\substack{\gamma \\ (\gamma_{j_1}=\cdots=\gamma_{j_{k}}=0)}}
    \left(\frac{\tau_1}{\tau_0}\right)^{p-\sum_{j}\gamma_j}
    \prod_{j}\left(\frac{1-\alpha_{j}}{\alpha_{j}}\right)^{1-\gamma_j}
    \frac{\sqrt{|\Sigma_{\gamma}|}}
    {r_{\gamma}^{a^*/2 + p/2}}
    \mathcal{T}_{a^*}(\mu_{\gamma}, \Sigma_{\gamma}),
\end{equation}
where $r_{\gamma} = \frac{y^Ty - y^T\mathbf{x}L_{\gamma}\mathbf{x}^Ty}{2} + b$
and $\Sigma_{\gamma} = \frac{a^*}{2r_{\gamma}}L^{-1}_{\gamma}$. Therefore,
\begin{align}
    P(\beta \mid y) 
    &\overset{\beta}{\propto} \sum_{\gamma}\left(
    \left(\frac{\tau_1}{\tau_0}\right)^{p - \sum_{j}\gamma_j}
    \prod_{j}\left(\frac{1-\alpha_j}{\alpha_j}\right)^{1-\gamma_j}
    \frac{\sqrt{|\Sigma_{\gamma}|}}{r^{a^*/2+p/2}_{\gamma}}
    \mathcal{T}_{a^*}(\mu_{\gamma}, \Sigma_{\gamma})
    \right)\\
    & = \frac{\sum_{\gamma}\left(
    \left(\frac{\tau_1}{\tau_0}\right)^{p - \sum_{j}\gamma_j}
    \prod_{j}\left(\frac{1-\alpha_j}{\alpha_j}\right)^{1-\gamma_j}
    \frac{\sqrt{|\Sigma_{\gamma}|}}{r^{a^*/2+p/2}_{\gamma}}
    \mathcal{T}_{a^*}(\mu_{\gamma}, \Sigma_{\gamma})
    \right)}
    {\sum_{\gamma}\left(
    \left(\frac{\tau_1}{\tau_0}\right)^{p - \sum_{j}\gamma_j}
    \prod_{j}\left(\frac{1-\alpha_j}{\alpha_j}\right)^{1-\gamma_j}
    \frac{\sqrt{|\Sigma_{\gamma}|}}{r^{a^*/2+p/2}_{\gamma}}\right)}\label{eq:joint:post:beta:gen}.
\end{align}
This shows that joint posterior of $\beta$ can be represented as a $2^p$ component
mixture of multivariate t-distribution, where each component corresponds
to a particular combination of selected variables out of $2^p$ possible combinations.

Now, by simplifying \cref{eq:joint:post:beta:gen}, we have the joint posterior of $\beta$
\begin{align}
    P(\beta \mid y)
    & = \frac{\sum_{\gamma}\left(
    \left(\frac{\tau_1}{\tau_0}\right)^{p - \sum_{j}\gamma_j}
    \prod_{j}\left(\frac{1-\alpha_j}{\alpha_j}\right)^{1-\gamma_j}
    \frac{\sqrt{|\Sigma_{\gamma}|}}{r^{a^*/2+p/2}_{\gamma}}
    \mathcal{T}_{a^*}(\mu_{\gamma}, \Sigma_{\gamma})
    \right)}
    {\sum_{\gamma}\left(
    \left(\frac{\tau_1}{\tau_0}\right)^{p - \sum_{j}\gamma_j}
    \prod_{j}\left(\frac{1-\alpha_j}{\alpha_j}\right)^{1-\gamma_j}
    \frac{\sqrt{|\Sigma_{\gamma}|}}{r^{a^*/2+p/2}_{\gamma}}\right)}\\
    & = \frac{\sum_{\gamma}\left(
    \left(\frac{\tau_1}{\tau_0}\right)^{p - \sum_{j}\gamma_j}
    \prod_{j}\left(\frac{1-\alpha_j}{\alpha_j}\right)^{1-\gamma_j}
    \frac{2^{p/2}r_{\gamma}^{p/2} (a^*)^{-p/2}
    \sqrt{|L_{\gamma}|}}{r^{a^*/2+p/2}_{\gamma}}
    \mathcal{T}_{a^*}(\mu_{\gamma}, \Sigma_{\gamma})
    \right)}
    {\sum_{\gamma}\left(
    \left(\frac{\tau_1}{\tau_0}\right)^{p - \sum_{j}\gamma_j}
    \prod_{j}\left(\frac{1-\alpha_j}{\alpha_j}\right)^{1-\gamma_j}
    \frac{2^{p/2}r_{\gamma}^{p/2} (a^*)^{-p/2}
    \sqrt{|L_{\gamma}|}}{r^{a^*/2+p/2}_{\gamma}}\right)}\\
    & = \frac{\sum_{\gamma}\left(
    \left(\frac{\tau_1}{\tau_0}\right)^{p - \sum_{j}\gamma_j}
    \prod_{j}\left(\frac{1-\alpha_j}{\alpha_j}\right)^{1-\gamma_j}
    \frac{\sqrt{|L_{\gamma}|}}{r^{a^*/2}_{\gamma}}
    \mathcal{T}_{a^*}(\mu_{\gamma}, \Sigma_{\gamma})
    \right)}
    {\sum_{\gamma}\left(
    \left(\frac{\tau_1}{\tau_0}\right)^{p - \sum_{j}\gamma_j}
    \prod_{j}\left(\frac{1-\alpha_j}{\alpha_j}\right)^{1-\gamma_j}
    \frac{\sqrt{|L_{\gamma}|}}{r^{a^*/2}_{\gamma}}\right)}\\
    &=\frac{\sum_{\gamma}\left(
    \left(\prod_{j} \alpha_j^{\gamma_j}
    (1-\alpha_j)^{1-\gamma_j}\right)\left(\frac{\sqrt{| L_{\gamma}|}}
    {\tau_1^{\sum\gamma_j} \tau_0^{(p-\sum\gamma_j)}}\right) 
    \frac{1}{\left(b + \frac{y^Ty - {\mu_{\gamma}}^T L_{\gamma}^{-1} \mu_{\gamma}}{2}\right)^{n/2+a}}
    \mathcal{T}_{a^*}(\mu_{\gamma}, \Sigma_{\gamma})
    \right)}
    {\sum_{\gamma}\left(
    \left(\prod_{j} \alpha_j^{\gamma_j}
    (1-\alpha_j)^{1-\gamma_j}\right)\left(\frac{\sqrt{| L_{\gamma}|}}
    {\tau_1^{\sum\gamma_j} \tau_0^{(p-\sum\gamma_j)}}\right) 
    \frac{1}{\left(b + \frac{y^Ty - {\mu_{\gamma}}^T L_{\gamma}^{-1} \mu_{\gamma}}{2}\right)^{n/2+a}}\right)}.
\end{align}

}

\bibliographystyle{plainnat}
\bibliography{basu20c}

\end{document}